\documentclass[12pt]{article}
\usepackage{graphicx}
\usepackage{amstext}
\usepackage{caption}
\usepackage{etoolbox}
\usepackage{makeidx}
\include{epsf}
\usepackage[toc,page]{appendix}
\usepackage{amsfonts}
\usepackage[numbers,sort&compress]{natbib}
\usepackage{authblk} 
\usepackage{sectsty}
\usepackage{amsmath,amssymb,epsfig}
\usepackage{amscd}
\usepackage{amsthm}
\usepackage{braket}
\usepackage{amsmath}
\usepackage{bm}
\usepackage{xcolor}
\usepackage{mathrsfs}
\usepackage{setspace}
\usepackage{url}
\usepackage{dsfont}
\usepackage[applemac]{inputenc}
\usepackage[english]{babel}
\usepackage{enumitem} 
\usepackage{verbatim}
\usepackage{slashed}
\usepackage[]{latexsym}
\usepackage{mathtools}
\usepackage{subcaption}
\usepackage{ mathrsfs }
\usepackage{braket}
\usepackage{caption}
\usepackage{hyperref}
\hypersetup{
pdftitle={},%
pdfauthor={},%
pdfsubject={},%
pdfkeywords={},%
colorlinks=true,%
linkcolor=blue,%
citecolor=red,%
linktocpage=true,%
pageanchor=true
}

\newcommand{\be}{\begin{equation}}
\newcommand{\ee}{\end{equation}}
\def\beqa{\begin{eqnarray}}
\def\eeqa{\end{eqnarray}}
\def\bean{\begin{eqnarray*}}
\def\eean{\end{eqnarray*}}
\newcommand{\bd}{\mathbf{d}}
\newcommand{\R}{\mathbb{R}}

\newcommand{\nn}{{\nonumber}}

\newcommand{\eqn}[1]{(\ref{#1})}
\newcommand{\del}{\partial}

\newtheorem{thm}{Theorem}[section]

\newtheorem{defi}{Definition}[section]

\renewenvironment{thebibliography}[1]
         {\section*{References}\frenchspacing\small
          \begin{list}{[\arabic{enumi}]}
         {\usecounter{enumi}\parsep=2pt\topsep 0pt
         \settowidth{\labelwidth}{[#1]}
         \leftmargin=\labelwidth\advance\leftmargin\labelsep
         \rightmargin=0pt\itemsep=1pt\sloppy}}{\end{list}}

 \numberwithin{equation}{section}

\title{
Topological and dynamical aspects of Jacobi sigma models }
\date{}

\author[1,2]{Francesco Bascone}
\author[1]{Franco Pezzella}
\author[1,2]{Patrizia Vitale}
\affil[ ]{}
\affil[1]{\textit{\footnotesize INFN - Sezione di Napoli, Complesso Universitario di Monte S. Angelo Edificio 6, via Cintia, 80126 Napoli, Italy.}}
\affil[2]{\textit{\footnotesize Dipartimento di Fisica ``E. Pancini'', Universit\`a di Napoli Federico II, Complesso Universitario di Monte S. Angelo Edificio 6, via Cintia, 80126 Napoli, Italy.}}
\affil[ ]{}
\affil[ ]{\footnotesize e-mail: \texttt{francesco.bascone@na.infn.it, franco.pezzella@na.infn.it, patrizia.vitale@na.infn.it}}

\begin{document}
\maketitle
\begin{abstract}The geometric properties of sigma models with target space a Jacobi manifold are investigated. In their basic formulation, these are topological field theories - recently introduced by the authors  - 
 which share and generalise relevant features of Poisson sigma models, such as gauge invariance under diffeomorphisms and finite dimension of the reduced phase space. 
 After reviewing the main novelties and peculiarities of these models, we perform a detailed analysis of constraints and ensuing gauge symmetries in the Hamiltonian approach. Contact manifolds as well as locally conformal symplectic manifolds are discussed, as main instances of Jacobi manifolds.

\end{abstract}

\section{Introduction}

The present paper is a follow-up of  \cite{Bascone2021} where a non-linear sigma model with target space a Jacobi manifold has been introduced. The aim of the paper is to review  the findings of \cite{Bascone2021} in order  to clarify some important points which where not addressed in detail in the preceding paper. Moreover, proofs of main results, such as the dimensionality of the reduced phase space of the model and  the characterisation and closure of the algebra of gauge transformations, are re-derived,   overcoming  some simplifying  assumptions which where previously made. 

The main motivation for the search of a consistent definition of a sigma model with target space a Jacobi manifold is certainly the fact that it represents a natural, non-trivial generalisation of the well known Poisson sigma model. 
The latter  is a topological field theory which was first introduced \cite{Ikeda1994,Schaller1994} in relation with   two-dimensional field theories with non-trivial target space, e.g. gauge  and gravity models, as well as gauged WZW models. 
One interesting feature of the model is its intimate relation with the  geometry of the target space. Indeed, it makes it possible to unravel mathematical aspects of such manifolds by employing techniques from field theory. An example of this relation was given by Cattaneo and Felder in \cite{Cattaneo2001,Cattaneo2001a} where they show that the reduced phase space of the Poisson sigma model is actually the symplectic groupoid integrating the Lie algebroid associated with the Poisson structure of the target manifold. Moreover, the  model made it possible to give an alternative  derivation of Kontsevich  quantisation formula for Poisson manifolds, in terms of the Feynman diagrams coming from the perturbative expansion of the field theory \cite{Cattaneo2000}. Analogous questions, such as the 
 geometry of the reduced phase space  and the quantisation of Jacobi structures, could be addressed once the model is understood.

From a more physical point of view, 
another motivation for the introduction of this new model is the perspective of applying techniques from Topological Quantum Field Theory to the analysis of new string backgrounds, as well as the possibility of  obtaining some useful description of known models  within the  framework of Jacobi manifolds, as it is the case for the Poisson setting.

The Poisson sigma model is described in terms of  fields $(X, \eta)$ which are formally associated with   a bundle map from the tangent bundle of a source space $\Sigma$, a two-dimensional orientable manifold possibly with boundary, to the cotangent bundle of the target Poisson manifold $M$. {In particular,} $X$ is the base map,  describing the embedding of $\Sigma$ into $M$, while $\eta$ is the fibre map, an auxiliary field which is in particular a one-form on $\Sigma$ with values in the pull-back of the cotangent bundle over $M$. In general it is not possible to integrate out such an auxiliary field, unless the target space is a symplectic manifold. In this case the Poisson bi-vector can be inverted and the equations of motion can be solved for $\eta$. The resulting action is that of a topological A-model \cite{Witten1988, Witten1998}, {i.e.,} $S=\int_{\Sigma} X^*(\omega)$, where $\omega=\Pi^{-1}$ is the symplectic form on $M$, $\Pi$ the Poisson bi-vector field (fulfilling  the condition of zero Schouten bracket $[\Pi, \Pi]_S=0$)  and $X^*$ denotes the pull-back map. 

Our aim in \cite{Bascone2021} was to investigate the possibility of  relaxing the condition $[\Pi, \Pi]_S=0$ to what is probably the  most  natural generalisation, represented by  a Jacobi structure.  
The latter is specified  by  a bi-vector  field  $\Lambda$  and a vector field $E$,  the so called Reeb vector field, satisfying 
\be
[\Lambda,\Lambda]_S= 2E\wedge\Lambda\;\;\;\, {\rm and} \;\;\;\,[E,\Lambda]_S=0.  
\ee
 The triple  $(M, \Lambda, E)$ defines  a Jacobi manifold.
A Poisson manifold is  a particular case with $E=0$ everywhere.  Two main families of  Jacobi manifolds, with all other cases being recovered as intermediate situations\footnote{It is possible to show (see for example \cite{Vaisman2002} Thm.  11) that a generic Jacobi manifold admits a   foliation by locally conformal symplectic and/or contact leaves. Examples of Jacobi manifolds with "nonpure"  characteristic foliation, namely with leaves  of odd and even dimension, i.e., contact and l.c.s. leaves may be found in \cite{deLeon1997}.},  are represented by contact and locally conformal symplectic manifolds, which we will consider later in the paper for applications of our model.

From a Jacobi structure one can construct Jacobi brackets on the algebra of functions on  $M$ with the following definition:
\begin{equation}
\{f, g\}_J=\Lambda(d f, d g)+f(E g)-g(E f).
\end{equation}
The latter satisfy the Jacobi identity, but unlike Poisson brackets, violate the Leibniz rule; in other words,  the Jacobi bracket still endows the algebra of functions on $M$ with a Lie algebra structure, but it is not a derivation of the point-wise product among functions. 
Thus,  the bi-vector field $\Lambda$ may be ascribed to the family of   bi-vector fields violating Jacobi identity, such as  "twisted" and "magnetic" Poisson structures (see for example \cite{Szabomag, Szabononass}) which recently received some interest in relation  with the quantisation of higher structures (their Jacobiator being non-trivial) and with the description of non-trivial geometric fluxes in string theory. The violation of Jacobi identity is however under control, because the latter is recovered by the full Jacobi bracket, which is alternatively defined as the most general
local bilinear operator on the space of real functions $ C^\infty(M,\R)$ which is skew-symmetric and
satisfies Jacobi identity \cite{deLeon1997}, and this makes its study especially interesting to us.

The Jacobi sigma model generalises the construction of the Poisson sigma model via the inclusion of an additional field on the source manifold, which is necessary in order to   take into account the  new background vector field $E$. The field variables  of the model are represented  by $(X, \eta, \lambda)$, where $X: \Sigma \to M$ is the usual embedding map, while $(\eta,\lambda)$ are put together to give elements of $\Omega^1(\Sigma,  X^*(T^*M\oplus \R))$, being $T^*M\oplus \R=J^1 M$ the vector bundle of 1-jets of real functions on $M$. The resulting theory is a two-dimensional topological non-linear gauge theory describing strings sweeping  a Jacobi manifold. The following main results {were} achieved in \cite{Bascone2021}:
\begin{itemize}
\item  Similarly to the Poisson sigma model, the reduced phase space can be proven to be  finite-dimensional, but while for the Poisson case the dimension is $2\text{dim}M$, for the Jacobi sigma model the dimension is $2\text{dim}M-2$.
\item The model  may be related to a Poisson sigma model with target space $M\times \R$ within a ``Poissonization" procedure. The latter approach has been pursued  in \cite{Chatzistavrakidis2020} in relation with non-closed fluxes, and \cite{Vancea:2020bwu} with reference to gauge symmetry. 
\item  The auxiliary fields $(\eta,\lambda)$ can be integrated out, both for contact and locally conformal symplectic manifolds so to get a model which is solely defined in terms of the field $X$ and its derivatives.
\item By including a dynamical term which is proportional to  the metric tensor of the target manifold, it is possible to obtain a Polyakov action. The  background metric and the $B$-field are expressed in terms of   the Jacobi structure.  A non-zero three form, $H= dB$, may occur, depending on the details of the model. 
\end{itemize}

As already anticipated, the main purpose of this contribution, is to give an account of the progress made in our understanding of the Jacobi sigma model, both for the  topological and the dynamical version. Since the matter is recent and likely to be further developed  we believe it appropriate to present a detailed, self-contained  review of the material already covered in \cite{Bascone2021}. 
  
The paper is organised as follows. In section \ref{Secpoissonsigmamodel} we present a short summary of the Poisson sigma model. In section \ref{Secjacobimanifolds} we review the notion of Jacobi manifold and Jacobi structure, and we describe the procedure of Poissonization of a Jacobi manifold $M$ yielding to a higher dimensional manifold $M \times \mathbb{R}$. Although strictly not necessary for the purposes of the present paper, the latter has played an important role in suggesting the original formulation of the model.   In section \ref{secjacobisigmamodel}
the action functional   for the Jacobi sigma model is stated. The model in the canonical formulation is  constrained, with first class constraints generating gauge transformations. In comparison with  \cite{Bascone2021} the analysis of first and second class constraints is considerably enlarged and clarified in \ref{dirco}.  
Gauge transformations are implemented by generating functionals $K_{\beta,\lambda_t}$ through Poisson brackets. By identifying the gauge parameters $(\beta,\lambda_t)$  with sections of the pullback  bundle $J^1M$ we will show that the algebra of gauge generators closes off-shell under milder assumptions than in \cite{Bascone2021}. To obtain this result, the notion of generalised Koszul bracket \cite{Vaisman2000,Kerbrat1993}, which extends to Jacobi manifolds the Koszul bracket defined on Poisson manifolds, has been used. In \ref{redph} the constrained phase space is reduced with respect to gauge symmetries and shown to be finite dimensional with dimension equal to 2dimM-2. This result, which is the content of Theorem \ref{mainthm}, was already presented in \cite{Bascone2021}. Here, thanks to a better understanding of the constrained phase-space, the proof of the theorem is improved and given for general field configurations\footnote{Specifically, the auxiliary fields  $\lambda_t$ and $\beta_i$ need not be related by the  condition $\beta_i= \del_i \lambda_t$ and $\lambda_u$ need not be zero, as invoked in \cite{Bascone2021}.}. 
 In section \ref{Secexamples} we consider the model in both cases of the target space being a contact and locally conformal symplectic manifold in a general fashion. We thus discuss in more detail than in \cite{Bascone2021} noteworthy examples, such as   the  manifolds  $SU(2)$ and $SU(2)\times S^1$ as instances of contact and LCS target spaces respectively. Finally,  section \ref{Secdynamical} contains a review of the dynamical case introduced in \cite{Bascone2021}. This consists in supplementing the action of the Jacobi sigma model with   a dynamical term which includes   a metric tensor on the target space. It is very much inspired to the dynamical Poisson sigma model discussed in \cite{Schupp2012}, with some interesting differences.  The emerging model, besides being   non-topological,  yields a Polyakov action with background metric $g$ and $B$-field determined by the Jacobi structure involved. We do not have new results in this respect and a complete understanding of the model is still lacking, while  we are presently working on it.
 We conclude with a final discussion of the results with remarks and perspectives.
 
\section{Poisson sigma model}\label{Secpoissonsigmamodel}

The Poisson sigma model is a two-dimensional topological field theory with target space a Poisson manifold, first introduced by Ikeda and independently by Schaller and Strobl in the context of two-dimensional gravity \cite{Schaller1994} and later widely investigated in relation with other models such as two-dimensional Yang--Mills and gravity theories, as well as in relation with deformation quantisation and branes. Related to that, we mention \cite{Cattaneo2000} by Cattaneo and Felder where  the Poisson sigma model is used to give a physical interpretation of Kontsevich quantisation formula  in terms of Feynman diagrams of the perturbative expansion of the model, and the work \cite{Cattaneo2001} by the same authors, in which they prove that the reduced phase space of the model is the symplectic groupoid integrating the Lie algebroid associated with the Poisson manifold, inspiring later works on the integrability of Lie algebroids \cite{Crainic2007}. A brief introduction to the topic can be found in \cite{Schallera}.

A smooth manifold $M$ is a Poisson manifold if there exists a bi-vector $\Pi \in \Gamma(\Lambda^2 TM)$ satisfying the Jacobi identity:
\begin{equation}
0=[\Pi, \Pi]_{S}^{i j k}=\Pi^{i \ell} \partial_{\ell} \Pi^{j k}+\operatorname{cycl}(i j k),
\end{equation}
where $[\cdot, \cdot]_S: \Lambda^{p}(M) \times \Lambda^{q}(M) \rightarrow \Lambda^{p+q-1}(M)$ is the Schouten--Nijenhuis bracket on the algebra of multivector fields on the manifold $M$. The Poisson bracket on $C^{\infty}(M)$ is then defined as $\{f,g \}=\Pi(df, dg), \, \, f, g \in C^{\infty}(M)$.

Let $\Sigma$ be a two-dimensional oriented manifold, possibly with boundary, and $(M, \Pi)$ {an} $m$-dimensional Poisson manifold. The topological Poisson sigma model is defined by the bosonic real fields $(X, \eta)$, with $X: \Sigma \rightarrow M$ the usual embedding map and $\eta \in \Omega^1(\Sigma, X^* (T^*M))$ a one-form on $\Sigma$ with values in the pull-back of the cotangent bundle over $M$. The action of the model is given by
\begin{equation}\label{actionpoissonsigma}
S=\int_{\Sigma}\left[\eta_{i} \wedge d X^{i}+\frac{1}{2} \Pi^{i j}(X) \eta_{i} \wedge \eta_{j}\right], \;\;\; i,j= 1,\dots, {\rm dim} M
\end{equation}
where $dX\in \Omega^1(\Sigma, X^* (TM))$ and the contraction of covariant and contravariant indices is relative to the 
pairing between differential forms on $\Sigma$  with values in $X^* (T^*M)$ and $X^* (TM)$, respectively. 
 It is induced
by the natural pairing between $T^*M$ and $TM$ and yields a two-form on $\Sigma$.  The action is manifestly invariant under diffeomorphisms of $\Sigma$, hence it describes a topological model.
In order to make the world-sheet dependence explicit in \eqn{actionpoissonsigma}, we introduce local coordinates  $u^{\mu}$ ($\mu=0,1)$ on $\Sigma$ so that $dX^i=\partial_{\mu} X^i du^{\mu}$, $\eta_i=\eta_{\mu i} du^{\mu}$ yielding
\begin{equation}
S= \int_{\Sigma} d^2u \left[\epsilon^{\mu \nu} \eta_{\mu i} \partial_{\nu} X^i+\frac{1}{2} \Pi^{ij}(X) \epsilon^{\mu \nu}\eta_{\mu i} \eta_{\nu j} \right].
\end{equation}

The variation of the action leads to the following equations of motion in the bulk:
\begin{equation}\label{eqmopoisson1}
dX^i+\Pi^{ij}(X)\eta_j=0,
\end{equation}
\begin{equation}\label{eqmopoisson2}
d\eta_i+\frac{1}{2}\partial_i \Pi^{jk}(X)\eta_j \wedge \eta_k=0.
\end{equation}
One thing to notice is that the consistency of the equations of motion requires $\Pi(X)$, as a background field, to satisfy the Jacobi identity \footnote{This can be understood by acting on \eqn{eqmopoisson1} with the exterior derivative, then using again \eqn{eqmopoisson1} and finally using \eqn{eqmopoisson2} on the result.}. 

If the manifold $\Sigma$ has a boundary, then it is necessary to impose  suitable boundary conditions such that the boundary term $\int_{\partial \Sigma} \delta X^i \eta_i$ vanishes. Many possibilities have been considered \cite{Falceto2010,Calvo2005,Calvo2006,Cattaneo2013}. Interestingly, these  have been associated with  different brane solutions when the Poisson sigma model is considered in the framework of topological string theory.   Indeed, taking the restriction of the field $X_{|\partial \Sigma}:\partial \Sigma \rightarrow N$, for some closed submanifold $N$ (the brane), there may be different conditions for $N$. The one usually chosen (in particular, it was used by Cattaneo and Felder in \cite{Cattaneo2001}), is the following: 
\begin{equation}\label{bcpoisson}
\eta(u)v=0 \, \, \, \, \forall \, v \in T(\partial \Sigma), \, \, u \in \partial \Sigma.
\end{equation}
The Poisson sigma model comprises a variety of  models. The most obvious is the one corresponding to a trivial Poisson structure $\Pi=0$, for which one simply has a BF model with action $\int_{\Sigma} \eta_i \wedge dX^i$. An interesting non-trivial example is the case corresponding to a linear Poisson structure $\Pi^{ij}={f^{ij}}_k X^k$, leading to a non-Abelian BF theory. In this case, in fact, the Jacobi identity for $\Pi$ makes $M$ the dual of a Lie algebra with structure constants ${f^{ij}}_k$, and $\eta$ takes the role of a one-form connection. Other cases are two-dimensional Yang-Mills theory (which is obtained by using a linear Poisson structure and including a Casimir function of $M$ as a non-topological term in the action), gauged Wess--Zumino--Witten models  and  two-dimensional gravity models. A useful review where all these models are considered as derived from the Poisson sigma model is \cite{Ikeda2017}. 

An important remark concerns  the auxiliary fields $\eta_i$, which encompass conjugate momenta of the configuration fields $X^i$ and Lagrange multipliers. On using the equations of motion they  can be integrated away, resulting in a second order action, only if the target space is a symplectic manifold. In this case, in fact, the Poisson bi-vector can be inverted to a symplectic form $\omega$, and the resulting action is that of the so-called A-model, with action $S=\int_{\Sigma} \omega_{ij} \,dX^i \wedge dX^j$. In the language of strings, this corresponds to a topological action with  $B$-field coinciding with the symplectic two-form.

We now focus on the Hamiltonian approach. Let us choose the topology of the world-sheet as $\Sigma=\mathbb{R} \times [0,1]$ (we are considering open strings), where we identify the local coordinates $(u_0, u_1)$ with time and space, respectively,   $u_0 = t \in \mathbb{R}$, $u_1 = u \in I= [0,1]$. By further denoting $\beta_i=\eta_{t i}$, $\zeta_i=\eta_{u i}$, $\dot{X}=\partial_t X$ and $X'=\partial_u X$, the first order Lagrangian can be written as
\begin{equation}\label{lagpo}
L(X, \zeta; \beta)=\int_I du \left[-\zeta_i \dot{X}^i+\beta_i\left(X'^i+\Pi^{ij}(X)\zeta_j \right) \right],
\end{equation}
from which it is clear that $X$ and $-\zeta$ are canonically conjugate variables, with Poisson brackets 
\begin{equation}\label{canpoi}
\{\zeta_i(u), X^j(v)\}=-{\delta_i}^j \delta(u-v),
\end{equation}
while all the other brackets are vanishing. 

Notice that, in this notation, the boundary condition \eqn{bcpoisson} means that $\beta_{|\partial I}=0$,  $\beta=\eta_t$ being the component of $\eta$ tangent to the boundary.

Since $\beta$ has no conjugate variable, it has to be considered as a Lagrange multiplier imposing the constraints
\begin{equation}\label{constraintpoisson}
X'^i+\Pi^{ij}(X) \zeta_j=0,
\end{equation}
from which it follows that the Hamiltonian 
\begin{equation}
H_{\beta}= -\int_I du \,\beta_i \left[X'^i+\Pi^{ij}(X)\zeta_j \right]
\end{equation}
is pure constraint and the constraint manifold $\mathcal{C}$ (the space of solutions of \eqn{constraintpoisson}) can also be understood as the common zero set of the functions $H_{\beta}$ for all $\beta$ satisfying the boundary condition $\beta(0)=\beta(1)=0$. 
This implies that  the system is invariant under time-diffeomorphisms. The infinitesimal generators are the Hamiltonian vector fields associated with $H_\beta$ by the canonical Poisson bracket \eqn{canpoi}
\be\label{gagen}
\xi_\beta= \{H_\beta, \cdot \}=\int du\, \left(\dot X^i \frac{\delta}{\delta X^i} + \dot \zeta_i \frac{\delta}{\delta \zeta_i}\right),
\ee
with
\begin{equation}\label{sympoisson1}
 \dot X^i=-\Pi^{ij}\beta_j,
\end{equation}
\begin{equation}\label{sympoisson2}
\dot  \zeta_i=\partial_u \beta_i-\partial_i \Pi^{jk}\zeta_j \beta_k \,. 
\end{equation}
The model is also invariant under space-diffeomorphisms $ f(u)\del_u $, the latter   being the Hamiltonian vector field associated with $H_\beta$ if one chooses $\beta_j = f(u) \zeta_j$ \cite{Cattaneo2001a}. However, in order for the algebra of generators to close, one has to extend the dependence of $\beta$ according to   $\beta (u)\rightarrow  \beta(u, X(u))$, with $\beta= \beta_i dX^i$ the associated one-form in local coordinates. Then  it is possible to check that
\be
\{H_\beta, H_{\tilde\beta}\}= H_{[\beta,\tilde\beta]}\label{gauge}
\ee
with 
\be
[\beta, \tilde\beta]= d\langle \beta,\Pi(\tilde\beta)\rangle - i_{\Pi(\beta)} d\tilde \beta + i_{\Pi(\tilde\beta)} d \beta \label{Kos}
\ee
 the Koszul bracket among one-forms on the target manifold $M$, which satisfies the Jacobi identity provided $\Pi$ is a Poisson tensor. $\langle\; ,\;\rangle$ denotes the natural pairing between $T^*M$ and $TM$.   Following \cite{Cattaneo2001}, Eq. \eqn{Kos} may be extended to $P_0\Omega^1(M)$, the latter being the algebra of continuous maps $\beta:I \rightarrow  \Omega^1(M)$, with the property {$\beta(0)= \beta(1)=0$}, according to
 \be\label{Kosexte}
 [\beta, \tilde \beta](u)=  [\beta(u), \tilde\beta(u)].
 \ee
Eq. \eqn{gauge} shows that the map $\beta\rightarrow H_\beta$ is a Lie algebra homomorphism, the Hamiltonian constraints are first class and the Hamiltonian vector fields \eqn{gagen} generate gauge transformations.
Hence, the reduced phase space of the model is defined in the usual way as the quotient $\mathcal{G}=\mathcal{C}/H$, where $H$ is the gauge group.
It can be proven \cite{Cattaneo2001} that the reduced phase space is a finite-dimensional manifold of dimension 2dim(M). A generalisation  of this result holds true for  the Jacobi sigma model, therefore we shall postpone the proof to a forthcoming section.

\section{Jacobi manifolds}
\label{Secjacobimanifolds}

Jacobi brackets were first introduced by Lichnerowicz in \cite{Lichnerowicz1978} as a natural generalisation of Poisson brackets, such that the Leibniz rule is replaced by a weaker condition. The brackets are defined by means of a bi-linear bi-differential operator acting on the algebra of functions on a smooth manifold $M$ as
\begin{equation}
\{f, g\}_J=\Lambda(d f, d g)+f(E g)-g(E f),
\end{equation}
where $\Lambda \in \Gamma(\Lambda^2 TM)$ is a {bi-vector} field and $E \in \Gamma (TM)$ is a vector field (the Reeb vector field) on the manifold $M$,  satisfying
\begin{equation}\label{jacobilambda}
[\Lambda, \Lambda]_S=2 E \wedge \Lambda, \quad [\Lambda, E]_S=\mathscr{L}_{E}\Lambda=0,
\end{equation}
where $\mathscr{L}$ denotes the Lie derivative operator.\footnote{In other words, the {bi-vector} field $\Lambda$  fails to satisfy Jacobi identity by a term  given by $E\wedge \Lambda$,  representing the Jacobiator.}
It will prove useful to write the explicit expression of these relations in coordinates: 
\begin{equation}\label{jacobiidentitygen}
\Lambda^{pi} \partial_p \Lambda^{jk}+{\rm cycl\, perm}\{ijk\}
=E^i \Lambda^{jk} +{\rm cycl\, perm}\{ijk\},
\end{equation}
\begin{equation}\label{liederivelambda}
E^k \partial_k \Lambda^{ij}-\Lambda^{kj} \partial_k E^i-\Lambda^{ik} \partial_k E^j=0.
\end{equation}
Jacobi brackets are linear, skew-symmetric and satisfy Jacobi identity just like Poisson brackets. The Leibniz rule  is replaced by the  condition
\begin{equation}
\{f, g h\}_J=\{f, g\}_J h+g\{f, h\}_J+g h(E f).
\end{equation}
This means that the Jacobi brackets still endow the algebra of functions $\mathcal{F}(M)$ with the structure of a Lie algebra \footnote{In particular a local Lie algebra, in the sense of Kirillov \cite{Kirillov1976}, since it holds that:  $\text {supp }(\{f, g\} )\subseteq\text {supp }( f) \cap \text {supp } (g)$.}, but, unlike the Poisson brackets,  they are not a derivation of the point-wise product among functions.
 Jacobi brackets are a generalisation of  Poisson brackets since the latter can be obtained from the former  if the Reeb vector field is vanishing, $E=0$.
A Jacobi manifold $(M, \Lambda, E)$ is then defined as a smooth manifold equipped with a Jacobi structure \footnote{A generalisation in terms of complex line bundles can be found in \cite{LV}.}.
Two main classes  of Jacobi manifolds are represented by locally conformal symplectic manifolds (LCS) and contact manifolds. 
The former ones are even-dimensional manifolds endowed with a non-degenerate two-form $\omega \in \Omega^2(M)$ and a closed one form $\alpha\in \Omega^1(M)$ with the property that \cite{Marle1991, Vaisman1985}
\be \label{lcscond}
d\omega+\alpha\wedge \omega=0. 
\ee
The latter condition is equivalent to the statement that  in each  local chart     $U_i\subset M$ there always exists $f\in C^\infty(U_i)$ such that $\alpha$ is exact and $\omega$ is conformally equivalent to a symplectic form $\Omega$ according to:
\be 
\alpha= df,\;\;\; \omega= e^{-f} \Omega.
\ee
Then, the  global structures   which qualify a LCS manifold as a Jacobi one, namely the pair $(\Lambda, E)$, are uniquely defined in terms of $(\alpha, \omega)$ as follows:
\begin{equation}\label{lcscond2}
\iota_E \omega=-\alpha, \quad \iota_{\Lambda(\gamma)}\omega=-\gamma \quad \forall \, \gamma \in T^*M.
\end{equation}
Globally conformal symplectic and symplectic  manifolds are particular cases of LCS. 

Contact manifolds are in turn odd-dimensional manifolds which are endowed with a  contact structure, namely a one-form $\vartheta$ satisfying $\vartheta \wedge (d\vartheta)^n =\Omega$, where $2n+1$ is the dimension of the manifold and $\Omega$ a volume form. 
The contact form is defined as an equivalence class of one-forms up to multiplication by a non-vanishing function. It is possible to endow the algebra of functions on a contact manifold with  a Lie algebra structure \cite{Asorey2017}, which reads
\begin{equation}\label{con1}
[f, g ]\vartheta \wedge(d \vartheta)^{n} : =(n-1) d f \wedge d g \wedge \vartheta \wedge(d \vartheta)^{n-1}+(f d g-g d f) \wedge(d \vartheta)^{n}.
\end{equation}
The latter is local by construction and satisfies Jacobi identity. It  may be seen to be equivalent to the standard   definition of Jacobi bracket, by implicitly defining $\Lambda$ and $E$ as follows:
\begin{equation}\label{contactcond}
\begin{aligned}
{} & \iota_E \left( \vartheta \wedge(d \vartheta)^{n}\right)=(d \vartheta)^{n} \\ &
\iota_{\Lambda} \left(\vartheta \wedge(d \vartheta)^{n}\right)=n \vartheta \wedge(d \vartheta)^{n-1}.
\end{aligned}
\end{equation}
The latter conditions imply\footnote{In three dimensions they are actually equivalent statements.} that
\begin{equation}\label{contactcond2}
\iota_E \vartheta=1, \quad \iota_{E} d\vartheta=0,
\end{equation}
as well as 
\begin{equation}\label{contactcond3}
\iota_{\Lambda} \vartheta=0, \quad  \iota_{\Lambda} d\vartheta=1.
\end{equation}

Noteworthy examples of contact manifolds are represented by three-dimensional semi-simple Lie groups, where the contact one-form can be chosen to be  one of the basis left-invariant (resp. right-invariant) one-forms on the group manifolds. Non-trivial examples of LCS manifolds may be easily constructed by considering the product  $M\times S^1$, with $M$ a contact manifold \cite{Vaisman1985}. In Section \ref{topjacsu2} we shall consider in some detail the case of   the Lie group $SU(2)$, which has been  been widely studied in the literature in relation with Poisson sigma models ({see, for example,}  \cite{Bonechi2005, Calvo2003}) and we shall only sketch its LCS counterpart $SU(2)\times S^1$.  

\subsection{Poisson structure on $M \times \mathbb{R}$ from $(M, J)$}
An important result due to Lichnerowicz  shows that it is possible to associate with any Jacobi manifold a higher dimensional Poisson manifold. The exact statement goes as follows
\begin{thm}\label{thpoissonization} \cite{Lichnerowicz1978}
{Given a Jacobi structure $J(f,g)=\Lambda(df, dg)+f(E g)-g(E f)$ on $M$, the product manifold $M \times \mathbb{R}$ carries a family of equivalent Poisson structures with  Poisson bi-vector $P$ defined as
\begin{equation}\label{poissonizat}
P \equiv e^{-\tau}\left(\Lambda+\frac{\partial}{\partial \tau} \wedge E\right), \;\;\;\;  \tau \in \mathbb{R}.
\end{equation}
}
\end{thm}
The procedure is called  Poissonization of the Jacobi manifold $(M, \Lambda, E)$. For contact manifolds it is possible to check that this is actually a symplectification. Indeed, in such a case, one can define a   closed $2$-form $\omega$ on $M \times \mathbb{R}$ in terms of the  contact one-form $\vartheta$: $\omega=d\left( e^{\tau} \pi^* \vartheta\right)=e^{\tau}\left(d\tau \wedge \pi^* \theta+d\pi^*\theta \right)$, where $\pi: M \times \mathbb{R} \to M$ is the projection map. By using the defining properties of the contact form, it is possible to check that $\omega$ is non-degenerate, hence symplectic.

The Poissonization procedure  provides a simple recipe to obtain a Poisson bracket from a Jacobi structure and can be used to derive useful results for Jacobi manifolds. By using the  projection map $\pi$  it is possible to define Hamiltonian vector fields associated with the Jacobi structure ({see, for example,} \cite{Vaisman2002}) 
\begin{equation}
\xi_f \coloneqq  \pi_*( \xi^P_{e^{\tau}f})|_{\tau=0},
\end{equation}
where $\xi^P_{e^{\tau}f}$ is the Hamiltonian vector field associated with the Poisson bracket on $M \times \mathbb{R}$ and $\pi_*:T(M \times \mathbb{R}) \rightarrow TM$ denotes the push-forward of the projection map. This yields,   for any function $f\in \mathcal{F}(M)$ 
\begin{equation}\label{Hamvec}
\xi_f=\Lambda(df, \cdot)+f E.
\end{equation}
The map $f\rightarrow \xi_f$ is homomorphism of  Lie algebras, it being  $[\xi_f, \xi_g]=\xi_{\{f,g\}_J}$, where the bracket $[\cdot, \cdot]$ is the standard Lie bracket of vector fields. 

Theorem \ref{thpoissonization} may be invoked to obtain  field theories on Jacobi manifolds by relying on existing models on the 
overlying Poisson manifold. This is achieved by  the immersion $i: M\hookrightarrow M\times \R$ through the identification of $M$ with $M \times \{0\}$. It applies in particular to the Poisson sigma model and it is the approach followed in the first part of \cite{Bascone2021} (also see \cite{Chatzistavrakidis2020, Vancea:2020bwu}). However, as we will see, it  is not the choice we have made for the  present   paper. 

\section{Jacobi sigma model}\label{secjacobisigmamodel}

{In this section,} we shall analyse the Jacobi sigma model, first   introduced in \cite{Bascone2021} (also see \cite{Chatzistavrakidis2020}) as a generalisation  of the Poisson sigma model. Although the defining action functional may be justified in terms of a Poissonization of the target   Jacobi manifold and further reduction of the correspondent Poisson sigma model living on $M\times \R$, it has been shown in \cite{Bascone2021} that an independent formulation can be given.  We shall adhere to the latter approach in this paper. Therefore the following coordinates-independent reformulation of the definition already given in \cite{Bascone2021} may be stated
\begin{defi}\label{defjacobiaction}
Let $\left(M, \Lambda, E \right)$ be a n-dimensional Jacobi manifold. The Jacobi sigma model with source space a two-dimensional manifold $\Sigma$ with boundary $\del \Sigma$ and target space $M$ is defined by the action functional 
\be
S[X,(\eta,\lambda)]= \int_\Sigma \langle\eta, ({ d} X)\rangle + \frac{1}{2}\langle\eta,(\Lambda\circ X)\eta\rangle + \lambda \wedge (E\circ X)\eta
\ee
with boundary condition  $\eta(u)v=0, u \in \del \Sigma, v\in T(\del \Sigma)$.
\end{defi}
The field  configurations are represented by $X, (\eta, \lambda)$ with   $X: \Sigma \to  M$ the base map and
$(\eta, \lambda ) \in \Omega^1(\Sigma,  X^*(J^1 M))$,
where $J^1 M=T^*M\oplus \R$ is the 1-jet bundle  of real functions on $M$.

Sections of the latter are isomorphic as a $C^\infty(M)$-module to the algebra of one-forms \cite{Vaisman2000} 
\be \Gamma_0(M):= \{e^\tau(\alpha+ f d\tau) | \alpha\in \Omega^1(M), f\in  C^\infty(M), \tau\in \R\} \subseteq \Omega^1(M\times \R) \label{gamma0}
\ee
 which is closed with respect to the Koszul bracket of the Poissonised manifold.
The map $\langle\;,\;\rangle$ establishes a pairing  between  differential forms on $\Sigma$ with values in the pull-back  $X^*(T^*M)$ and differential forms on $\Sigma$ with values in $X^*(TM)$. It is induced   by  the natural one between  $T^*M$ and $TM$ and yields in this case a two-form on $\Sigma$.  Then the action may be rewritten as (cfr. \cite{Bascone2021}) 
\begin{equation}\label{jacobiaction}
S(X, \eta, \lambda)=\int_{\Sigma} \left[\eta_i \wedge dX^i+\frac{1}{2}\Lambda^{ij}(X)\eta_i \wedge \eta_j-E^i(X) \eta_i \wedge \lambda \right].
\end{equation}
On comparing with the action of the Poisson sigma model \eqn{actionpoissonsigma} one  important difference  is the presence of a new auxiliary field, $\lambda$, which, loosely speaking, is  a one-form on the source manifold $\Sigma$ but a scalar on the Jacobi manifold. This is a consequence of the fact that the Jacobi bracket is expressed in terms of a bi-differential operator, not a bi-vector field. Therefore $\lambda$   is needed  in order to take into account the presence of the Reeb vector field $E$.

The variation of the action, together with the boundary condition for $\eta$ in Def. \ref{defjacobiaction}, gives the following equations of motion
\begin{equation}\label{eomjacobi1}
dX^i+\Lambda^{ij}\eta_j-E^i \lambda=0,
\end{equation}
\begin{equation}\label{eomjacobi2}
d\eta_i+\frac{1}{2}\partial_i \Lambda^{jk}\eta_j \wedge \eta_k-\partial_i E^j \eta_j \wedge \lambda=0,
\end{equation}
\begin{equation}\label{eomjacobi3}
E^i \eta_i=0.
\end{equation}
The boundary condition for $\eta$ ensures the vanishing of  boundary terms.
Consistency of the three yields another dynamical equation. In fact, on applying the exterior derivative to Eq.  \eqn{eomjacobi1} we obtain
\begin{equation}\label{interm}
\partial_k \Lambda^{ij} dX^k \wedge \eta_j+\Lambda^{ij}d\eta_j-\partial_k E^i dX^k \wedge \lambda-E^i d\lambda=0.
\end{equation}
By substituting Eqs. \eqn{eomjacobi1}-\eqn{eomjacobi3}  and by using the properties of a Jacobi structure, Eqs.~\eqn{jacobilambda},  we finally get 
\begin{equation}\label{eomjacobi4}
d\lambda = \frac{1}{2}\Lambda^{ij} \eta_i \wedge \eta_j.
\end{equation}

\subsection{Canonical formulation of the model}
In this section we will focus on the Hamiltonian formulation of the model, in close analogy with the procedure followed for the Poisson sigma model in Sec. \ref{Secpoissonsigmamodel}.  
To this, the source manifold is chosen to be  $\Sigma=\mathbb{R} \times [0,1]$, with  local coordinates $t\in \R$, $u\in [0,1]$.
Moreover, by explicitly indicating the time and space components,  the one-forms $dX, \eta$ and $\lambda$ shall be locally represented as  $dX=\dot{X}dt+X' du$, $\eta=\beta dt+ \zeta du$, 
$\lambda=\lambda_t dt+ \lambda_u du$, with $\lambda_t, \lambda_u$ scalar fields, while $\dot X, X'$ and $\beta, \zeta$ carrying and extra index on (the pull-back of)  the target manifold $M$.  Note that the boundary condition {in definition \ref{defjacobiaction}} results in $\beta_{\del\Sigma}=0$, just like for the Poisson sigma model, while there is no boundary condition for  $\lambda$ deriving from the variation of the action. We shall discuss  this  issue  later. 

With the notation chosen, the Lagrangian of the model acquires the form 
\begin{equation}\label{jacobilagrangian}
L=\int_I du \left[-\dot{X}^i \zeta_i+\beta_i\left(X^{'i}+\Lambda^{ij}\zeta_j -E^i \lambda_u\right)+\lambda_t\left( E^i \zeta_i \right) \right],
\end{equation}
with equations of motion
\beqa\label{expleoms}
\dot X^i&=& -\Lambda^{ij} \beta_j+E^i \lambda_t\nonumber\\
\dot \zeta_i &=& \beta^{'}_i-\partial_i \Lambda^{jk} \beta_j \zeta_k -\partial_i E^j \zeta_j \lambda_t+ \partial_i E^j \beta_j \lambda_u,
\eeqa
\beqa
\label{eulerlagrangeconstraints}
X'^i+\Lambda^{ij}\zeta_j-E^i \lambda_u=0 \nonumber\\
E^i \zeta_i=0 \nonumber\\
E^i \beta_i=0.
\eeqa
The evolutionary equations are, therefore, represented by  Eqs. \eqn{expleoms}, involving time derivatives, while Eqs.  \eqn{eulerlagrangeconstraints} represent constraints. In the following we perform a detailed analysis of the emergence and nature of constraints in the Hamiltonian approach.

\subsubsection{Dirac analysis of constraints}\label{dirco}
From the Lagrangian \eqn{jacobilagrangian} the Hamiltonian is seen to be 
\be\label{Hamiltonian}
H_0= - \int_I du \, \beta_i\left(X^{'i}+\Lambda^{ij}\zeta_j -E^i \lambda_u\right)+\lambda_t\left( E^i \zeta_i \right),
\ee
with $\pi_i= \delta L/\delta \dot X^i=-\zeta_i$ the conjugate momenta for the field $X^i$, while the conjugate momenta of all other fields are zero. The theory is therefore constrained. We shall perform the analysis \`a la Dirac, referring to standard textbooks for a detailed description of the procedure. \footnote{Shortly, we recall that primary constraints are those which emerge from the Lagrangian, without using the equations of motion. They identify a submanifold of the original carrier  space of the dynamics, $\mathcal{C} _1 \subset \mathcal{C}_0$. Secondary constraints are all subsequent constraints, obtained by the request that primary constraints be preserved along the motion. They  further constrain the motion to some $\mathcal{C} _2 \subset \mathcal{C}_1$. The process is iterated by imposing conservation of new constraints (tertiary, $\cdots$, n-ary, $\cdots$)  at each step, until the true manifold of the motion, $\mathcal{C}_n\subset \mathcal{C}_{n-1}\subset\cdots \subset\mathcal{C}_0 $, is found. The term "secondary constraints" is then used for all, except for primary  constraints.  

Dirac classification of constraints is yet another one, which is specific of the Hamiltonian setting \cite{Dirac1964}. Here the carrier space of the dynamics is phase space, endowed with a Poisson bracket. At each step of the reduction from the unconstrained phase space $ \mathcal{C}_0$, the so called  na\"ive Hamiltonian $H_0$ is replaced by a new one, say $H_i=H_0 + a_\mu \chi_\mu+ b_\mu \mathcal{G}_\mu$, with $\{\chi_\mu\}$ the primary constraints, and $\{\mathcal{G}_\mu\}$ the secondary constraints which have emerged up to the step $i$. The parameters $a_\mu, b_\mu$ are  also referred to as Lagrange multipliers. The process ends when all constraints, say $\psi_\mu$, are conserved, namely $\dot\psi_\mu= \{\psi_\mu, H_n\}\simeq 0$ on the constrained manifold $\mathcal{C}_n$.   On considering the Poisson algebra of all constraints, first class (primary and secondary) are those which close a  subalgebra, i.e. $\{\psi_\mu, \psi_\nu\}= f_{\mu\nu}^\kappa \psi_k \simeq 0$, whereas second class constraints obey $\{\psi_\mu, \psi_\nu\}= c_{\mu\nu} $, with $c_{\mu\nu}$ a non-degenerate matrix (second class constraints are therefore in even number). Because of that,  their Lagrange multipliers, say $d_\mu$, may be completely determined according to $d_\mu= -c_{\mu\nu} \{\psi_\nu, H_0\}$ as opposed to first class ones, which are left undetermined, hence,   give rise to gauge ambiguities. }
\be\label{Constraints}
\pi_{\beta_i}=0,\;\; \pi_{\lambda_u}= 0, \;\; {\pi_{\lambda_t}}= 0
\ee
which have to be added to the Hamiltonian $H_0$.
 The unconstrained phase space of the model may be identified as the infinite-dimensional manifold  $T^*P(M\times \R^{m}\times \R\times \R)$ with $P(N)$ denoting the space of maps from the source space $I=[0,1]$ to some target $N$. The configuration fields will be $X^i: I\rightarrow M, \, \beta_i: I\rightarrow \R^m, i= 1\dots m$, and $\lambda_t, \lambda_u: I\rightarrow \R$. It is possible to read off the non-zero Poisson brackets from the first order action, which  yields
\be\label{canon}
\{\pi_i(u), X^j(v)\}=\delta_i^j \delta(u-v)
\ee
to which we have to add those related with the extended phase space
\beqa
\{\pi_{\beta_i}(u), \beta_j(v)\}&=&\delta_i^j \delta(u-v)\\
\{\pi_{\lambda_t}(u), \lambda_t(v)\}&=& \delta(u-v)\\
\{\pi_{\lambda_u}(u), \lambda_u(v)\}&=& \delta(u-v).
\eeqa
By imposing that primary constraints be preserved along the motion, $m+2$ new constraints are obtained 
\be\label{constraints}
\begin{array}{lllll}
 \dot \pi_{\beta_i}&=& X'^i+ \Lambda^{ij}\zeta_j -E^i \lambda_u&:=&\mathcal{G}_{\beta_i}\\
\dot \pi_{\lambda_t}&=& E^i\zeta_i&:=&\mathcal{G}_{\lambda_t}\\
\dot \pi_{\lambda_u}&=& E^i\beta_i&:= &\mathcal{G}_{\lambda_u}. 
 \end{array}
 \ee
 Hence, the initial Hamiltonian $H_0$ is itself a sum of constraints
 \be
 H_0= -\int du\,  \left[ \beta_i \mathcal{G}_{\beta_i} + \lambda_t \mathcal{G}_{\lambda_t} \right]
 \ee
 Let us compute their Poisson algebra. For  secondary constraints we find
 \beqa
 \{\mathcal{G}_{\beta_i}(u), \mathcal{G}_{\beta_j}(v)\}&=& -\Lambda^{il}\frac{\del}{\del X^l(u)}\mathcal{G}_{\beta_j}(v) +\Lambda^{jl}\frac{\del}{\del X^l(v)}\mathcal{G}_{\beta_i}(u)  \label{sec1}\\
  \{\mathcal{G}_{\beta_i}(u), \mathcal{G}_{\lambda_u}(v)\}&=& -\Lambda^{il}\frac{\del}{\del X^l(u)} \mathcal{G}_{\lambda_u}(v)\label{sec2}\\
 \{\mathcal{G}_{\beta_i}(u), \mathcal{G}_{\lambda_t}(v)\}&=& -\Lambda^{il}\frac{\del}{\del X^l(u)}\mathcal{G}_{\lambda_t}(v)+ E^l \frac{\del}{\del X^l(v)}\mathcal{G}_{\beta_i}(u) \label{sec3}\\
  \{\mathcal{G}_{\lambda_t}(u),\mathcal{G}_{\lambda_u}(v)\}&=& -E^l \frac{\del}{\del X^l(u)}\mathcal{G}_{\lambda_u}(v)
  \label{sec4}\eeqa
 Before proceeding further, we assume, without loss of generality, that  a basis of vector fields on the target manifold $M$ has been chosen such that the Reeb vector field has non-zero component only along one of the basis elements, say $E^i= {\mathcal E}\delta^{im}$ and  we shall indicate with $a=1,\dots m-1$ the remaining directions. Thus we compute the remaining brackets, which yield
 \beqa\label{alco}
 \{\pi_{\beta_i}(u), \mathcal{G}_{\lambda_u}(v)\}&=& {\mathcal E}\,\delta^{im}  \delta(u-v) \label{alco1}\\
 \{\pi_{\lambda_u}(u), \mathcal{G}_{\beta_i}(v)\}&=&- {\mathcal E} \,\delta^{im}  \delta(u-v) \label{alco2})
 \eeqa
 with all other brackets strongly zero. By repeated use of Eqs.   \eqn{jacobiidentitygen}, \eqn{liederivelambda} in the chosen parameterization for the Reeb vector field as $E^i= \mathcal{E} \delta^{im}$,
  a tedious but straightforward calculation  gives an explicit expression for the Poisson brackets   \eqn{sec1}-\eqn{sec4}, which read 
 \beqa
  \{\mathcal{G}_{\beta_a}(u), \mathcal{G}_{\beta_b}(v)\}&=& \left[\mathcal{G}_{\beta_l}\del_l \Lambda^{ba}-
   \mathcal{G}_{\lambda_t} \Lambda^{ba}\right] \simeq 0  \label{secc1}\\
    \{\mathcal{G}_{\beta_a}(u), \mathcal{G}_{\lambda_t}(v)\}&=& 0 \label{secc2}\\
    \{\mathcal{G}_{\beta_a}(u), \mathcal{G}_{\beta_m}(v)\}&=& \left[\mathcal{G}_{\beta_l}\del_l  \Lambda^{ma}-\mathcal{G}_{\lambda_t} \Lambda^{ma}-\mathcal{E} \Lambda^{ak}\zeta_k\right] \simeq  - \mathcal{E} \Lambda^{ak}\zeta_k \label{secc3}\\
  \{\mathcal{G}_{\beta_m}(u), \mathcal{G}_{\lambda_t}(v)\}&=& \mathcal{G}_{\beta_l}\del_l {\mathcal E} - \mathcal{E}\delta'(u-v)\simeq - \mathcal{E}\delta'(u-v) \label{secc4}\\
  \{\mathcal{G}_{\beta_a}(u), \mathcal{G}_{\lambda_u}(v)\}&=&-\mathcal{G}_{\lambda_u}\del_m \Lambda^{am}\simeq 0\label{secc5}\\
    \{\mathcal{G}_{\beta_m}(u), \mathcal{G}_{\lambda_u}(v)\}&=&-  \beta_m \Lambda^{ml} \del_l  \mathcal{E}\label{secc6}\\
  \{\mathcal{G}_{\lambda_t}(u),\mathcal{G}_{\lambda_u}(v)\}&=& -\mathcal{E}\beta_m\del_m \mathcal{E} = - \mathcal{G}_{\lambda_u} \del_m\mathcal{E} \simeq 0.
  \label{secc7}
  \eeqa
The chosen parametrisation for the Reeb vector field is particularly useful because  it considerably simplifies the classification of constraints as first or second class. By inspecting the rank of the matrix of Poisson brackets,   it is easy to verify that the latter is always equal to four. {Therefore, we} may conclude that  four out of $2m+4$ constraints are second class, i.e.,
\be
\begin{array}{cc}
\pi_{\lambda_u}&\pi_{\beta_m}\\
\mathcal{G}_{\lambda_u}&\mathcal{G}_{\beta_m}
\end{array}.
\ee
By evaluating the conservation of constraints with respect to the total Hamiltonian 
 \be
 H_1= H_0+ \int du \, [a_i \pi_{\beta_i}+ a_t\pi_{\lambda_t}+ a_u\pi_{\lambda_u}]
\ee
 we may verify that no new constraints arise, but some of the Lagrange multipliers get fixed, namely 
\be
a_{m}= \beta_m=0, \;\;\; a_u= \beta_a \Lambda^{ak}\zeta_k + \del_u \lambda_t.
\ee
\footnote{As for the Lagrange multiplier $a_u$, we notice that its value agrees with the equation of motion for $\lambda_u$ which has been derived in the Lagrangian formalism, \eqn{eomjacobi4}.} The remaining $2m$ constraints, 
\be
\begin{array}{ccc}
\pi_{\beta_a},&\mathcal{G}_{\beta_a},& a=1, \dots,m-1\\
\pi_{\lambda_t},&\mathcal{G}_{\lambda_t}&
\end{array}
\ee
are first class, thus generating   gauge transformations, with generating functional given by the linear combination
\be\label{Kgen}
K(\beta_a, \lambda_t, a_t,a_{\beta_a})= \int du \, \lambda_t  \mathcal{G}_{\lambda_t} +\beta_a \mathcal{G}_{\beta_a} + a_t \pi_{\lambda_t}+a_{\beta_a} \pi_{\beta_a},\;\;\; a=1,\dots,m-1
\ee
and  $\beta_a, \lambda_t, a_t, a_{\beta_a}$ {are gauge parameters}\footnote{First class constraints have zero Poisson brackets with the total Hamiltonian, with undetermined Lagrange multipliers. Hence, they generate canonical symmetries, that is to say, gauge transformations. One main difference with respect to the Poisson sigma model is that for the latter the whole Hamiltonian is a first class constraint, hence being itself the generating function of gauge transformations. Here instead, the Hamiltonian contains second class constraints as well, which have to be subtracted in order to get the gauge generators.}. 

In order to compute the algebra of gauge generators, $\{K(\beta,\lambda_t, a_t, a_{\beta_a}), K(\tilde\beta,\tilde\lambda_t, \tilde a_t, \tilde a_{\beta_a})\}$,  we notice firstly  that  primary constraints  in  \eqn{Kgen} may be ignored, because their Poisson brackets are strongly zero. Secondly, it is evident from Eq. \eqn{secc1} that, similarly to the Poisson sigma model,  the algebra will only close on-shell.  Therefore, in order to obtain a closed algebra off-shell we  allow for the relevant gauge parameters to be functions of the fields. More precisely, given $(\beta_a, \lambda_t) \in C(I\rightarrow X^*( T^*M\oplus \R))$ we allow for  $\beta_a= \beta_a(u, X(u)), \lambda_t= \lambda_t(u, X(u))$. 
Thus, we compute
\beqa
\{K(\beta, \lambda_t), K(\tilde\beta, \tilde\lambda_t)\}&=&\int du du'  \left[ \{(\beta_a\mathcal{G}_a)(u), (\tilde\beta_b\mathcal{G}_b)(u')\}+ \{(\beta_a \mathcal{G}_a)(u), (\tilde\lambda_t \mathcal{G}_t)(u')\} \right. \nonumber\\
&+& \left. \{(\lambda_t \mathcal{G}_t)(u), (\tilde\beta_b \mathcal{G}_b)(u')\}+\{(\lambda_t \mathcal{G}_t)(u), (\tilde\lambda_t \mathcal{G}_t)(u')\} \right] . 
\eeqa
On using \eqn{canon} , where $\zeta_i= -\pi_i$ 
we find
\beqa
\{(\beta_a\mathcal{G}_a)(u), (\tilde\beta_b\mathcal{G}_b)(u')\}&=& \left[\mathcal{G}_c \left(\beta_a\tilde\beta_b \del_c \Lambda^{ba}-\beta_a \Lambda^{aj}\del_j\tilde\beta_c+ \tilde\beta_a \Lambda^{aj}\del_j\beta_c \right)\right.\nn\\
& -& \left. \mathcal{G}_t \beta_a\tilde\beta_b \Lambda^{ba}\right]\delta(u-u')
\\
\{(\beta_a \mathcal{G}_a)(u), (\tilde\lambda_t \mathcal{G}_t)(u')\}&=&\left(\mathcal{G}_a  \mathcal{E}\tilde\lambda_t \del_m \beta_a-\mathcal{G}_t\beta_a\Lambda^{aj}\del_j \tilde\lambda_t \right)\delta(u-u')\\
\{(\lambda_t \mathcal{G}_t)(u), (\tilde\beta_b \mathcal{G}_b)(u')\}&=&\left( \mathcal{G}_t \tilde\beta_b\Lambda^{bj}\del_j \lambda_t - \mathcal{G}_b \mathcal{E}\lambda_t \del_m \tilde\beta_b \right)\delta(u-u')\\
\{(\lambda_t \mathcal{G}_t)(u), (\tilde\lambda_t \mathcal{G}_t)(u')\}&=& \mathcal{G}_t\mathcal{E}(\tilde\lambda_t\del_m \lambda_t  - \lambda_t\del_m \tilde \lambda_t ) \delta(u-u').
\eeqa
This yields
\beqa\label{poibKK}
&&\{K(\beta, \lambda_t), K(\tilde\beta, \tilde\lambda_t)\}=\int du du' \left[\mathcal{G}_c\Bigl(\beta_a\tilde\beta_b \del_c \Lambda^{ba} +\Lambda^{aj}(\tilde\beta_a \del_j\beta_c-\beta_a \del_j\tilde\beta_c)  \Bigr.\right.\nonumber \\
&&\;\;\;\;\;\;\left.\Bigl.+\mathcal{E}\left(\tilde\lambda_t\del_m\beta_c-\lambda_t\del_m\tilde\beta_c\right)\Bigr)\right. \nn\\
&&\;\;\;\;\;\; \left.+\mathcal{G}_t \left( \beta_a\tilde\beta_b  \Lambda^{ab}+\Lambda^{aj}(\tilde \beta_a \del_j\lambda_t -\beta_a \del_j\tilde\lambda_t )+ \mathcal{E}\left(\tilde\lambda_t\del_m \lambda_t -\tilde\lambda_t\del_m \tilde\lambda_t\right)\right) \right]
\eeqa
We now observe that a generalisation of the Koszul bracket \eqn{Kos} is available for Jacobi manifolds, which endows the set of   sections of  the 1-jet bundle $ J^1M$  with a Lie algebra structure \cite{Kerbrat1993,Vaisman2000}.  
Given $(\alpha, f), (\beta, g)$ sections of $ J^1M$, namely $\alpha,\beta\in \Omega^1(M), f,g\in C(M)$, the bracket reads\footnote{Vaisman shows in \cite{Vaisman2000} that this is nothing but the Koszul bracket \eqn{Kos} defined for the associated ``Poissonized" manifold $(M\times\R, P)$, with respect to which the algebra of sections of  $J^1M$ is closed.}
\beqa\label{vaibra}
[(\alpha, f), (\beta, g)]&= &\Big( \bigl ( L_{\sharp_\Lambda \alpha } \beta-L_{\sharp_\Lambda \beta} \alpha  - d(\Lambda(\alpha,\beta)+ f L_E \beta -g L_E \alpha-\alpha(E) \beta+\beta(E)\alpha)\bigr) ,\Big.\nonumber \\
&& \Big. 
 \bigl(\{f,g\}_J -\Lambda(df -\alpha, dg -\beta\bigr)\Bigr)
\eeqa
where $\sharp_\Lambda \alpha$ denotes the vector field obtained by contracting the {bi-vector} field $\Lambda$ with the one-form $\alpha$; in local coordinates it reads: $\sharp_\Lambda \alpha=\alpha_i \Lambda^{ij}\del_j $. The latter satisfies Jacobi identity, provided the manifold is a Jacobi manifold, with $\{f,g\}_J$  the Jacobi bracket. 
 Analogously to the Poisson sigma model,   Eq. \eqn{vaibra} may be extended to   maps from the interval $I$ to sections of the 1-jet bundle $(\alpha, f):I \rightarrow  \Gamma(J^1M)$, with the property $\alpha(0)= \alpha(1)=0$, according to
 \be\label{Kosext}
 [(\alpha, f), (\beta, g)](u)=  [(\alpha, f)(u), (\beta, g)(u)].
 \ee
On computing the bracket \eqn{vaibra} for $(\beta, \lambda_t), (\tilde\beta, \tilde\lambda_t)$ a lengthy but straightforward calculation yields
\be[(\beta, \lambda_t), (\tilde\beta, \tilde\lambda_t)]=(\underline{\boldsymbol\beta}, \underline{\boldsymbol\lambda}_t)
\ee
with 
\beqa
\underline{\boldsymbol\beta}&=& \Bigl(\Lambda^{ij}(\beta_i \del_j\tilde\beta_k -\tilde\beta_i \del_j\beta_k)+ \tilde\beta_i\beta_j\del_k \Lambda^{ij}+\mathcal{E}(\lambda_t\del_m\tilde\beta_k-\tilde\lambda_t\del_m\beta_k)\nn\\
&+& \mathcal{E}(\tilde\beta_m \beta_k -\beta_m\tilde\beta_k)+(\lambda_t\tilde\beta_m - \tilde\lambda_t \beta_m)\del_k{\mathcal E} \Bigr)dX^k \label{Gamma}\\
\underline{\boldsymbol\lambda}_t&=& \Lambda^{ij}(\beta_i \del_j\tilde\lambda_t - \tilde\beta_i\del_j\lambda_t -\beta_i\tilde\beta_j)+ \mathcal{E}(\lambda_t\del_m\tilde\lambda_t-\tilde\lambda_t\del_m\lambda_t)\label{Gammat}
\eeqa
 Therefore, by  taking into account the second class constraints, which enforce $\beta_m=\tilde\beta_m=0$,   the RHS of the Poisson bracket \eqn{poibKK} may be stated in terms of \eqn{Gamma},\eqn{Gammat} to give
\be\label{KKrel}
\{K_{(\beta, \lambda_t)}, K_{(\tilde\beta, \tilde \lambda_t)}\}=- K_{[(\beta,\lambda_t), (\tilde\beta, \tilde\lambda_t)]}
\ee
Notice that, for $\beta_a =\del_a \lambda_t$ and analogous expression for $\tilde\beta_a$, the latter further reduces to
\be\label{poialge}
\{K_{(\beta, \lambda_t)}, K_{(\tilde\beta, \tilde \lambda_t)}\}= -K_{(d\{\lambda_t,\tilde\lambda_t\}_J, \{\lambda_t,\tilde\lambda_t\}_J)}
\ee
which is the particular case considered in \cite{Bascone2020}. The mapping $f\rightarrow e^\tau(df + f d\tau)$ with $f\in C^\infty(M)$ is a Lie algebra homomorphism from the Jacobi algebra of $M$ to $\Gamma_0(M)$ defined in \eqn{gamma0}. 

To summarise, the model exhibits first class constraints, which generate gauge transformations. Differently from the Poisson sigma model, second class constraints are present, which have to be dealt with, before analysing  the algebra of gauge generators. Thanks to the  bracket \eqn{vaibra} the  map
$
(\beta,\lambda_t) \rightarrow K(\beta,\lambda_t) 
$
is a Lie algebra homomorphism. Moreover, because of the homomorphism stated at the end of last paragraph, time-space diffeomorphisms may be  explicitly related with the Hamiltonian vector fields associated with $\lambda_t$ (resp. $\lambda_u$) through the Jacobi bracket (see \cite{Bascone2021} for details). 

It is to be noticed that, because of the presence of second class constraints, the  Hamiltonian vector fields generating infinitesimal gauge transformations are not directly associated with the Hamiltonian, but  rather with the functional $K_{(\beta, \lambda_t)}$. They shall  be explicitly worked out in the forthcoming section.

\subsubsection{The reduced phase space}\label{redph}

The reduced phase space of the model was  proven to be finite-dimensional in \cite{Bascone2021} with the simplifying assumption that $\lambda_u$ be zero. Moreover, a detailed analysis of constraints was not performed and  the  structure of the constrained manifold  was not completely clarified. Therefore, we shall repeat in what follows the proof of finite dimensionality of the reduced phase space in full generality.

  Since the model is gauge invariant under the action of the gauge transformations generated by the flows of the Hamiltonian vector field associated with the functional $K$, we can define the reduced phase space as the quotient space $\mathcal{C}/H$, where $H$ is the gauge group and $\mathcal{C}$ is  the constraint manifold. According to  Sec.  \ref{dirco} the former is an infinite-dimensional manifold, which after the imposition of all constraints results to be labelled by $2m$  fields. We  choose to parametrise the manifold with  $ X^i, \zeta_a, \lambda_u$.
The quotient manifold  $\mathcal{C}/H$ is finite-dimensional. Indeed the following theorem holds.
\begin{thm}\label{mainthm}
Let $(X^i, \zeta_a, \lambda_u) \in \mathcal{C}$. The subspace of $T_{(X^i, \zeta_a, \lambda_u)}\mathcal{C}$ spanned by the Hamiltonian vector fields $\xi_{\beta, \lambda_t}$ is a closed subspace of codimension $2 \text{dim}(M)-2$.
\end{thm}
\begin{proof}
Let us consider the subspace $\mathcal{S}_{(X^i, \zeta_a, \lambda_u)}$ of $T_{(X^i,\zeta_a, \lambda_u)}\mathcal{C}$ spanned by the Hamiltonian vector fields $\xi_{\beta, \lambda_t}$, associated with the functional $K(\beta, \lambda_t)$, generating infinitesimal gauge transformations. 
The map 
$(\beta, \lambda_t) \rightarrow \xi_{\beta, \lambda_t}$, explicitly given by 
\beqa
\delta_{\xi_{K} }X^i & \coloneqq& \{K(\beta, \lambda_t), X^i \}=\Lambda^{ia} \beta_a-\mathcal{E}\delta_m^i\lambda_t \label{trasf1}\\
\delta_{\xi_{K} } \zeta_a &\coloneqq& \{K(\beta, \lambda_t), \zeta_i \}=-(\beta_a)'+\beta_b \partial_a \Lambda^{bk}  \zeta_k+ \lambda_t  \zeta_m  \partial_a \mathcal{E}\label{trasf2}\\
\delta_{\xi_{K} } \zeta_m &\coloneqq& \{K(\beta, \lambda_t), \zeta_m \}=\beta_b\partial_m \Lambda^{bk}  \zeta_k+ \lambda_t \zeta_m  \partial_m {\mathcal E}  \label{trasf3}\\
\delta_{\xi_{K} } \lambda_u &\coloneqq& \{K(\beta, \lambda_t), \lambda_u \}=0 \label{trasf4}
\eeqa
 is linear. However, on the constraint manifold $\mathcal{C}$, the last terms in the r.h.s. of \eqn{trasf2} and \eqn{trasf3} vanish because $\zeta_m=0$, moreover, $\partial_m \Lambda^{bk}  \zeta_k = \partial_m \Lambda^{bc }  \zeta_c$ which  {is zero because of eq.} \eqn{jacobilambda}. Therefore, the  non-zero components of the map on the constraint manifold  are given by
\beqa
\xi_1^i 
&=&\Lambda^{ia} \beta_a-\mathcal{E}\delta_m^i\lambda_t  \label{hamvectcomp1}\\ 
 \xi_{2,a} 
 &=&-(\beta_a)'+\beta_b \partial_a \Lambda^{bk}  \zeta_k  \label{hamvectcomp2}
\eeqa
The  kernel of this linear map is empty, showing  that the map is injective. To this, we have to impose that Eqs. \eqn{hamvectcomp1}, \eqn{hamvectcomp2}  be zero. 
The second condition yields a homogeneous linear first order ODE with initial condition $\beta(0)=0$, {hence,} $\beta$ vanishes identically. The first one is instead an algebraic relation for which, by using the solution $\beta=0$ we have $\mathcal{E}\delta_m^i \lambda_t=0$ and since the Reeb vector field is nowhere vanishing it has to be $\lambda_t=0$. Hence, the map is injective.

{Let us, therefore,}  consider the image space.  The tangent vector $(\tilde{X}^i, \tilde{\zeta}_a), $ to a point $(X^i, \zeta_a, \lambda_u) \in \mathcal{C}$ is the solution of the linearised constraint equations
\begin{equation}\label{linearizedconstr}
\tilde{X}'^i+\left({A_j}^i -\partial_j \mathcal{E}\delta^{im} \lambda_u\right)\tilde{X}^j+\Lambda^{ib}\tilde{\zeta}_b =0.
\end{equation}
where we defined ${A_j}^i=\partial_j \Lambda^{ik}\zeta_k$. The tangent field has  no component $\tilde \lambda_u$ because of the constraint $\mathcal{G}_{\beta_m}$. 
If $(\tilde{X}, \tilde{\zeta})$ is an Hamiltonian vector field, and thus it is in the image of $\xi$, then it has to be 
\begin{equation}
\tilde{X}^i=\Lambda^{ib} \beta_b-\mathcal{E} \delta^{im} \lambda_t,
\end{equation} 
\be\label{proofeq1}
\tilde{\zeta}_a=-(\beta_a)'+{A_a}^b \beta_b.
\ee
The former have to hold at each $u$, which implies in particular $(\tilde{X}, \tilde{\zeta})$ is in the image of $\xi$ if 
\begin{equation*}
\tilde{X}^i(0)+\mathcal{E} (X(0))\delta^{im}\lambda_t(0)=0.
\end{equation*} 
If we introduce the matrix $V=\hat{P}\exp[-\int A\, du ]$ as the path-ordered exponential of $A$, i.e. the solution of the differential equation
\begin{equation}\label{defvmatrix}
\begin{cases}
(V_i^j)'=-V_i^k(u) {A_k}^j(u) \\
V_i^j(0)=\delta_i^j,
\end{cases}
\end{equation}
then Eq. (\ref{proofeq1}) can be rewritten in the form
\begin{equation}
\tilde{\zeta}_a (u)=-(V^{-1}(u))_a^c \, \partial_u[V(u) ^b_c\beta_b (u)].
\end{equation}
From this equation we can define the $m-1$ functions
\begin{equation}
p_a(u) \coloneqq \int_0^u dv V(v)^b_a \tilde{\zeta}_b(v)=- \int_0^u \partial_v[V(v) ^b_a\beta_b (v)],
\end{equation}
from which it follows that 
\begin{equation*}
\int_I du \, V(u)_a^b \tilde{\zeta_b}(u)=0.
\end{equation*}
Hence, we conclude that if $(\tilde{X}, \tilde{\zeta})$ is in the image of $\xi$, then we have 
\begin{equation}\label{conditionsproofdim}
\tilde{X}^i(0)+\mathcal{E}\delta^{im} (X(0))\lambda_t(0)=0, \quad \ \quad  \int_I du \,  V(u)_a^b \tilde{\zeta_b}(u)=0.
\end{equation} 
Now, it is important to notice that these conditions yield   $2m-2$ invariants and not $2m-1$ as it appears. Indeed, in the chosen parametrisation for the Reeb vector field, the first equation in  \eqn{conditionsproofdim} amounts to
\begin{equation}
\tilde{X}^a(0)=0, \quad \tilde{X}^m(0)=-\lambda_t(0).
\end{equation}
However, the second relation is not gauge invariant and  {does not} fix the $m-th$ component of $\tilde X$,  $\lambda_t(0)$ not being fixed to assume any   particular value. Therefore, the first of Eqs. \eqn{conditionsproofdim} yields $m-1$ invariant conditions. The final count of invariant conditions is then $2m-2$.

Viceversa, if we now consider $(\tilde{X}, \tilde{\zeta})$ as a tangent vector at the point $(X, \zeta, \lambda_u) \in \mathcal{C}$ satisfying  Eq. (\ref{conditionsproofdim}), then we show that this tangent vector is  Hamiltonian if we choose  $\beta_a=-(V^{-1})^b_a \, p_b=-(V^{-1})^b_a \int_0^u dv \, V(v)^{c}_b \tilde{\zeta}_{c}(v)$. To verify the statement, let us  define the vector field 
\begin{equation}\label{defYfield}
Y^i(u)=\Lambda^{ib}(u) \beta_b (u) - \mathcal{E}(u) \delta^{im} \lambda_t(u),
\end{equation}
satisfying the boundary condition  $Y^i(0)=-\mathcal{E}(0)\delta^{im} \lambda_t (0)$,  with the choice
\begin{equation}\label{defbetaproof}
\beta_a=-(V^{-1})^b_a \int_0^u dv \, V(v)^{c}_b \tilde{\zeta}_{c}(v).
\end{equation}
We will now check directly that $Y$ satisfies the same ODE as $\tilde X$ with the same boundary condition, namely it is a tangent vector field. The derivative of  Eq. \eqn{defYfield} with respect to $u$ yields:
\begin{equation*}
\begin{aligned}
Y'^i {} & =-\partial_k \Lambda^{ib} X'^k (V^{-1})^c_b \int_0^u dv V_c^{a} \tilde{\zeta}_{a}-\Lambda^{ib}\left[\partial_u (V^{-1})^c_b \int_0^u dv V_c^a \tilde{\zeta}_a+(V^{-1})^c_b V_c^a \tilde{\zeta}_a \right] \\ & +\partial_k \mathcal{E}\delta^{im} X'^k \lambda_t+\mathcal{E}\delta^{im} \lambda'_t.
\end{aligned}
\end{equation*}
By means of Eq. \eqn{eomjacobi4} with $\dot\lambda_u=\{K(\beta,\lambda_t), \lambda_u\}=0$, namely  $\lambda'_t=-\Lambda^{ij} \beta_i \zeta_j$ and  the constraint equation $X'^i=-\Lambda^{ib} \zeta_b+\mathcal{E}\delta^{im} \lambda_u$ we arrive at
\bean
Y'^i &=& \beta_b \zeta_c \left( \Lambda^{kc} \partial_k \Lambda^{ib}+\Lambda^{ik}\partial_k \Lambda^{cb} -\mathcal{E}\delta^{im} \Lambda^{bc}\right)- \Lambda^{ib} \tilde{\zeta}_b\\
&-&\partial_k \mathcal{E}\delta^{im}\Lambda^{kb} \zeta_b \lambda_t-\mathcal{E}\partial_m \left(\Lambda^{ib} \beta_b \lambda_u-\mathcal{E}\delta^{im}\lambda_t \lambda_u\right).
\eean
where we have substituted   
the defining equation for $V$ (\ref{defvmatrix}) and the explicit form  of  $\beta$ Eq. (\ref{defbetaproof}) .
Now using the Schouten bracket  (\ref{jacobiidentitygen}) we obtain
\begin{equation*}
Y'^i=-\partial_k \Lambda^{ib} \zeta_b Y^k-\Lambda^{ib}\tilde{\zeta}_b+\partial_k \mathcal{E}\delta^{im} \lambda_u Y^k+\left(\partial_m \Lambda^{bi} \mathcal{E}-\partial_k \mathcal{E}\delta^{im} \Lambda^{kb} \right) \left(\zeta_b \lambda_t+\beta_b \lambda_u \right).
\end{equation*}
Further implementing  $\mathscr{L}_E \Lambda=0$  we have finally
\begin{equation*}
\left(\partial_m \Lambda^{bi} \mathcal{E}-\partial_k \mathcal{E}\delta^{im} \Lambda^{kb} \right) \zeta_b \lambda_t=0 .
\end{equation*}

The same can be shown for the last term proportional to $\lambda_u$, so we have finally that $Y$ satisfies the linearised constraint in Eq. (\ref{linearizedconstr}) with the same boundary condition.

To conclude, we have proven that the image of $\xi$ is the subspace spanned by $\xi_{\beta, \lambda_t}$ modulo the $2m-2$ conditions (\ref{conditionsproofdim}), i.e. it is a closed subspace of co-dimension $2m-2$.
\end{proof}
Therefore,  similarly to the Poisson sigma model, the constraint manifold quotiented by gauge transformations results to be finite-dimensional, but of dimension equal to $2m-2$, with $m$ the dimension of the target Jacobi manifold.

\subsection{Poissonization}

In this section we review the almost  one-to-one correspondence between the Jacobi sigma model described in the previous sections and the reduced model which may be  obtained on the Jacobi manifold after Poissonisation. 

The idea in \cite{Bascone2021} was  to formulate a Poisson sigma model with target  $(M \times \mathbb{R}, P)$ $P$ being the Poisson tensor described in  Theorem \ref{thpoissonization}, and project its dynamics  
down to $M$.
 Fig. \ref{fig:diagram} illustrates  schematically   the procedure.
\begin{center}
\begin{figure}[ht]
\centering
\includegraphics[width=0.5\linewidth]{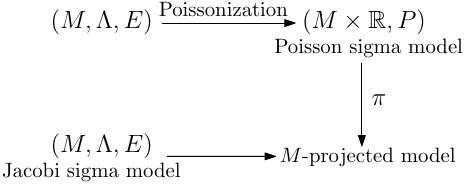}
\caption{Diagrammatic summary of the reduction of the dynamics from the Poisson sigma model to the Jacobi sigma model.}
\label{fig:diagram}
\end{figure}
\end{center}
For this purpose, let  us consider the Poisson sigma model with target space the Poisson manifold $\left(M\times \mathbb{R}, P \right)$ and  Poisson structure $P=e^{-X_0}\left(\Lambda+\frac{\partial}{\partial X_0} \wedge E\right)$ defined in terms of the structures of the embedded Jacobi manifold  $(M,\Lambda,E)$  and $X_0 \in \mathbb{R}$,   according to  Theorem \ref{thpoissonization}. The field configurations are then $(X, \eta)$,  with $X^{I}=(X^i, X^0) : \Sigma \to M \times \mathbb{R}$ the usual embedding maps  and $\eta \in \Omega^1(\Sigma, X^*(T^*(M \times  \mathbb{R})))$,   $\eta_{I}=(\eta_i, \eta_0)$. Capital indices  $I =0, \cdots,m$  label  coordinates over  the Poisson manifold $M \times \mathbb{R}$, while lowercase letters $i=1, \cdots,m$  shall be reserved to  the Jacobi manifold $M$. The Poisson bi-vector field in a coordinate basis {$\{\partial/\partial X^I \}$} can be written explicitly as
\begin{equation}
P^{I J} = e^{-X_0}\begin{pmatrix}
& & & & -E^1\\
& & & & \\
&  & \Lambda^{ij} & &\vdots \\
& & & & \\
& & & & -E^m\\
E^1 & & \cdots &  E^m & 0
\end{pmatrix},
\end{equation}
with $P=P^{I J}\partial_{I} \wedge \partial_{J}$ and $E=E^i \partial_i$. 

By splitting the equations of motion, (\ref{eqmopoisson1}) and (\ref{eqmopoisson2}) in terms of target coordinates adapted to the product manifold, one obtains:
\begin{equation}\label{eqmodecomp1}
dX^i+e^{-X^0}\left(\Lambda^{ij}\eta_j-E^i \eta_0 \right) =0
\end{equation}
\begin{equation}\label{eqmodecomp2}
dX^0+e^{-X^0} E^i \eta_i=0
\end{equation}
\begin{equation}\label{eqmodecomp3}
d\eta_i+\frac{1}{2} e^{-X^0} \partial_i \Lambda^{jk}\eta_j \wedge \eta_k+e^{-X^0}\partial_i E^j \eta_0 \wedge \eta_j=0
\end{equation}
\begin{equation}\label{eqmodecomp4}
d\eta_0-\frac{1}{2} e^{-X^0} \Lambda^{jk}\eta_j \wedge \eta_k-e^{-X^0} E^j \eta_0 \wedge \eta_j=0.
\end{equation}
We now consider the  immersion  {$i: M\hookrightarrow M\times \R$} through the identification  of $M$ with $M \times \{0\}$.   The reduced dynamics on $M$ is thus obtained     by posing $X^0=0$. This yields 
\begin{equation}\label{eqmopoissonization}
\begin{aligned}
{} & dX^i+\Lambda^{ij}\eta_j-E^i \eta_0=0 \\ &
E^i \eta_i=0 \\ &
d\eta_i+\frac{1}{2}\partial_i \Lambda^{jk}\eta_j \wedge \eta_k+\partial_i E^j \eta_0 \wedge \eta_j=0 \\ &
d\eta_0-\frac{1}{2}\Lambda^{jk}\eta_j \wedge \eta_k=0.
\end{aligned}
\end{equation}
On  identifying  $\eta_0$ with $\pi^* \lambda$, $\pi:M\times\R \rightarrow M$  being the projection map, it is immediate to verify that the reduced dynamics coincides with the one obtained from the action functional in Eq. \eqn{jacobiaction}.

However, it is important to remark that the two models are not completely equivalent. In fact, the reduced sigma model inherits  an additional boundary condition for the field $\lambda_t$ which comes from    $\eta_{0|\del\Sigma}=0$ for the Poissonised sigma model.  More precisely, upon performing the splitting of the world-sheet as $\Sigma = \mathbb{R} \times I$ the additional boundary condition requires $\lambda_{t|\del I}=0$ other than $\beta_I=0$, a condition which is   unnecessary for   the model described by the action  \eqn{jacobiaction}. This makes a difference in the analysis of the gauge transformations of the two models, although it is always possible to add this condition by hand if one wants to recover a complete equivalence between the two models. 

\section{Contact and LCS manifolds}
\label{Secexamples}

In this section we will consider in some detail two main classes  of target spaces for the Jacobi sigma model, that is contact and locally conformally symplectic manifolds. As a first application, we shall show that for both cases an interesting result can be stated, which concerns the possibility  of integrating out the auxiliary momenta and obtain a second order formulation of the action functional, solely expressed in terms of the embedding maps $X^i$ and their derivatives.

In general, it is not possible to integrate the auxiliary fields away so to obtain a second order action. As we already recalled in Sec. \ref{Secpoissonsigmamodel}, for the Poisson sigma model this  is possible only when the target space is a symplectic manifold. In that case the Poisson bi-vector can be inverted  and the equations of motion can be solved for $\eta$. We shall see in the following that the situation is different for the Jacobi sigma model,  both for  contact and LCS target.

\subsection{Integration on contact manifolds}\label{intcont}

Let us start by considering $M$ as a $(2n+1)$-dimensional contact manifold with contact one-form $\vartheta$ satisfying $\vartheta \wedge (d\vartheta)^n \neq 0$ at every point. The Jacobi structure can then be obtained from \eqn{contactcond}, or, equivalently, \eqn{contactcond2}-\eqn{contactcond3}.

Let us consider the equations of motion,   represented by \eqn{eomjacobi1}-\eqn{eomjacobi3}, \eqn{eomjacobi4}. Thanks  to the relations satisfied by the contact form, Eqs. \eqn{contactcond2}, \eqn{contactcond3} ,  the former  can be solved for $\eta$ and $\lambda$. In fact, on  multiplying \eqn{eomjacobi1} by $\vartheta_i$ and summing, we  obtain  
\begin{equation}
\vartheta_i(dX^i+ \Lambda^{ij}\eta_j- E^i \lambda)=\vartheta_i dX^i-\lambda=0,
\end{equation}
from which
\begin{equation}\label{contactaux1}
\lambda=\vartheta_i dX^i.
\end{equation}
 In order to obtain  $\eta$ we multiply   \eqn{eomjacobi1}  by $(d\vartheta)_{\ell i}$, and sum over $i$. Again, using the properties of the contact form we find
\begin{equation}
{d\vartheta}_{\ell i}(dX^i+\Lambda^{ij}\eta_j-E^i \lambda)=(d\vartheta)_{\ell i } dX^i+\delta_\ell^j\eta_j=0,
\end{equation}
from which we obtain 
\begin{equation}\label{contactaux2}
\eta_i=(d\vartheta)_{ij} dX^j.
\end{equation}
{Thus, we may} conclude  that the auxiliary fields can be completely integrated away.  Substituting \eqn{contactaux1}-\eqn{contactaux2} back into the action \eqn{jacobiaction} we find the following second order action
\begin{equation}\label{secondordercontact}
S_2=-\frac{1}{2}\int_{\Sigma} (d\vartheta)_{ij} \, dX^i \wedge dX^j= -\frac{1}{2}\int_{\Sigma} X^*(d\vartheta)
\end{equation}
where in the second equality we have restored  the pull-back map in order to highlight the geometric content.  The exterior derivative of the contact one-form takes the role of the $B$-field, which turns out to be closed for contact manifolds.  Despite the analogy with the symplectic case, the latter can only be non-degenerate when appropriately restricted to submanifolds of the target space. 
\subsubsection{Topological Jacobi sigma model on $SU(2)$}
\label{topjacsu2}
As a main example of the model described so far, we consider the  target space to be the group manifold of $SU(2)$, bearing in mind that  the  procedure may be adapted to  any three-dimensional semisimple Lie group. The group manifold is diffeomorphic to the sphere $S^3$. The contact one-form may be chosen among the basis   left-invariant (resp. right-invariant) one-forms of the group, say $\theta^i$ defined through the Maurer--Cartan one-form $\ell^{-1}d\ell=\theta^i e_i \in \Omega^1(SU(2), \mathfrak{su}(2))$,   with $\ell\in SU(2)$, $e_i = i \sigma_i/2$ the Lie algebra generators and $\sigma_i$ the Pauli matrices. Let us   choose, to be definite, the contact one form to be $\vartheta=\theta^3$. The latter defines  a Jacobi bracket according to Eq. \eqn{con1}
 it being
\be
\vartheta\wedge d\vartheta = \Omega
\ee
with $\Omega=\theta^1\wedge\theta^2\wedge \theta^3$ the volume form on the group manifold. {Therefore,} the Reeb vector field and the {bi-vector} field $\Lambda$ are easily determined by solving the equations
\beqa
&&\iota_E \vartheta=1, \quad \iota_{E} d\vartheta=0,\\
&&\iota_{\Lambda} \vartheta=0, \quad  \iota_{\Lambda} d\vartheta=1.
\eeqa
We obtain
\be\label{jacobisu2}
E= Y_3\;\;\; \Lambda = Y_1\wedge Y_2
\ee
with $Y_i, i=1,..,3$ the left invariant vector fields on the group manifold, which are dual the the one-forms $\theta^i$ by definition. Hence, the Reeb vector field is constant and orthogonal to the distribution spanned by the bi-vector field $\Lambda$. 
The action functional of the model is given by
\be
S[\phi,(\eta,\lambda)]= \int_\Sigma \langle\eta, \phi^*(g^{-1} \bd g)\rangle + \frac{1}{2}\langle\eta,(\Lambda\circ \phi)\eta\rangle + \lambda \wedge (E\circ \phi)\eta
\ee
with field configurations $\phi, (\eta, \lambda)$, $\phi: \Sigma \ni (t,u) \rightarrow g \in G$ and $(\eta, \lambda) \in \Omega^1(\Sigma, \phi^*(T^*G \oplus \mathbb{R}))$. We have chosen in this specific example to distinguish the exterior derivative $\bd$ on  the target manifold from the one on the source, $d$. We recall the  boundary condition $\eta(u)v=0, u \in \del \Sigma, v\in T(\del \Sigma)$.

The map $\langle\;,\;\rangle$ establishes a pairing  between  differential forms on $\Sigma$ with values in the pull-back  $\phi^*(T^*G)$ and differential forms on $\Sigma$ with values in $\phi^*(TG)$. 

On identifying the tangent space $TG$ with $G\times \mathfrak{g}$ and $T^*G$ with $G\times \mathfrak{g}^*$  we may write
\be
\phi^*(g^{-1} \bd g)= (g^{-1}\del_t g)^i e_i dt+ (g^{-1}\del_u g)^i e_i du = A^i (t,u) e_i dt + J^i (t,u) e_i du
\ee
where we have introduced the notation 
\be \label{currents}
(g^{-1}\del_t g)^i= A^i, \;\;\;\; (g^{-1}\del_u g)^i =J^i
\ee
 with 
$\{e_i\}$ a basis in the Lie algebra.  Analogously 
\beqa
\eta&=& \eta_{tj} e^j dt + \eta_{uj} e^j du := \beta_j e^j dt +\zeta_j e^j du \\
\lambda&= &\lambda_t dt + \lambda_u du
\eeqa
with $\eta_{tj} =\beta_j$, $\eta_{uj}= \zeta_j$ and $\{e^i\}$ a dual basis in $\mathfrak{g}^*$. Then the action is rewritten as 
\beqa
S[g,(\eta,\lambda)]&=& \int_\Sigma \eta_i \wedge \phi^*(g^{-1} dg)^i+ \frac{1}{2}\Lambda^{ij}\eta_i \eta_j  + \lambda \wedge E^i\eta_i\nonumber\\
&= & \int_\Sigma \left(\beta_i J^i -\zeta_i A^i + \Lambda^{ij} \beta_i \zeta_j +\lambda_t E^j \zeta_j -\lambda_u E^j \beta_j \right) du \, dt
\eeqa
and we have renamed  the map $\phi$ with $g$, to simplify the notation.

Let us now derive the equations of motion. By varying the action with respect to the fields $\zeta, \beta, g, \lambda_t, \lambda_u$ we find
\beqa
A^j &=& -\Lambda^{jl}\beta_l +\lambda_t E^j  \label{su2jacobi1}\\
J^j &=& -\Lambda^{jl}\zeta_l +\lambda_u E^j \label{su2jacobi2}\\
\del_t \zeta_j&=&-(\beta_k J^l -\zeta_k A^l) c^k_{lj} +\del_u \beta_j\label{su2jacobi3}\\
E^j\zeta_j&=&E^j\beta_j=0\label{su2jacobi4}
\eeqa
where we have used, to derive  the third equation,
\beqa
(\delta J)^j&=& (g^{-1}\del_u g)^l (g^{-1}\delta g)^k c^j_{lk} + \del_u (g^{-1}\delta g)^j  \\
(\delta A)^j&=& (g^{-1}\del_t g)^l (g^{-1}\delta g)^k c^j_{lk} + \del_t (g^{-1}\delta g)^j 
\eeqa
and $c^j_{lk}$ are the structure constants of the Lie algebra $\mathfrak{su}(2)$. Let us notice that, with the parameterisation chosen for the source manifold $\Sigma$, the evolutionary equations are the first and the third one, involving time derivatives, whereas the others are constraints. 

In order to make contact with Eqs. \eqn{eomjacobi1}-\eqn{eomjacobi3} previously derived for a generic  target space, we may write Eqs. \eqn{su2jacobi1}-\eqn{su2jacobi4} in compact form 
\beqa
\phi^*(g^{-1}\bd g)^j + \Lambda^{jl}\eta_l -\lambda E^j&=&0 \label{jasu21}\\ 
d\eta_j + \eta_k\wedge\phi^*(g^{-1}\bd g)^l c^k_{lj}&=&0\label{jasu22}\\
E^j \eta_j&= &0.
\eeqa
The first and last one match respectively   Eqs. \eqn{eomjacobi1}, \eqn{eomjacobi3}, once we have identified $X^i$ with the local coordinates describing the map $\phi$ in a  chart. The second equation needs an intermediate step: we   obtain  $\phi^*(g^{-1}\bd g)^j $ from \eqn{jasu21} and replace it in \eqn{jasu22}. We find
\be\label{su2consist}
d\eta_j + \eta_k\wedge \left(-\Lambda^{lm}\eta_m +\lambda E^l\right) c^k_{lj}=0
\ee
Then we observe that 
\be
\Lambda^{lm} c^k_{lj}= \frac{1}{2}( \mathscr{L}_{Y_j} \Lambda)^{mk}\;\;\; {\rm and} \;\;\;E^l c^k_{lj}=-(\mathscr{L}_{Y_j} E)^k
\ee
so that Eq. \eqn{su2consist} becomes
\be
d\eta_j + \frac{1}{2}( \mathscr{L}_{Y_j} \Lambda)^{km}\eta_k\wedge \eta_m  -(\mathscr{L}_{Y_j} E)^k\eta_k\wedge \lambda =0
\ee
and this is exactly Eq. \eqn{eomjacobi2}.

The Lagrangian may also be recast in the following form
\be
L[g,\eta,\lambda]=\int_I  du \left[ -\zeta_i A^i + \beta_i \left( J^i+\Lambda^{ij} \zeta_j- \lambda_u E^i\right)+\lambda_t E^j \zeta_j  \right] 
\ee
with $A^i$ playing now the role of velocities. 
The action is already in its first order form, with Hamiltonian
\be
H_0=- \int du \, \left[ \beta_i \left( J^i-\Lambda^{ij} \pi_j- \lambda_u E^i\right)+\lambda_t E^j \zeta_j \right]
\ee
and 
\be
\pi_i =\frac{\delta L}{\delta A^i}=-\zeta_i
\ee
being the only non-zero  momenta, whereas 
\be
\pi_{\beta_i}=\pi_{\lambda_t}=\pi_{\lambda_u}=0.
\ee   
The latter are primary constraints, which we add to the Hamiltonian to get
\be
H_1=-\int du \, \left[  \beta_i \left( J^i-\Lambda^{ij} \pi_j- \lambda_u E^i\right)+\lambda_t E^j \zeta_j + a_u \pi_{\lambda_u}+ a_t \pi_{\lambda_t} + a_{\beta_i}\pi_{\beta_i} \right].
\ee
In view of performing the Dirac analysis of constraints, the unconstrained phase space of the model may be identified as  the infinite-dimensional manifold $T^*(P(G\times \R^m\times \R\times\R))$, with $PM$ denoting the space of maps from the source space $\Sigma$ to the target manifold $M$. The  configuration fields will be $g: \Sigma\rightarrow G$, \; $\beta_i:\Sigma\rightarrow \R^m, i=i\dots m$ and $\lambda_u, \lambda_t : \Sigma\rightarrow \R$ . 
Then, we read off the {non-zero} Poisson brackets from the canonical one-form 
\be
\Theta=\int_I  du \;  \pi_i  \phi^*(g^{-1} {\bf d} g)^i 
\ee
and its exterior derivative 
\be
\Omega= { d}\Theta= \int_I {d} \pi_i  \wedge \phi^*(g^{-1} {\bf d} g)^i - \pi_i c^i_{jk} \phi^*(g^{-1} {\bf d} g)^j\wedge \phi^* (g^{-1} {\bf d} g)^k
\ee
This yields the non-zero Poisson brackets to be
\beqa
\{\pi_i(u), \pi_j(v)\}&=& c_{ij}^k \pi_k \delta(u-v)\label{piPB}\\
\{\pi_i(u), g(v)\}&=&  i g \sigma_i \delta(u-v)\\
\{g(u), \tilde g(v)\}&=&0
\eeqa
(in particular $\{\pi_i, J^j\}= J^k {c_{ki}}^j \delta(u-v) {+ \delta_{i}^j \delta'(u-v)}
$
), 
to which we add those on the extended phase space
\beqa
\{\pi_{\beta_i}(u), \beta_j(v)\}&=& \delta^{i}_j\delta(u-v)\\
\{\pi_{\lambda_t}(u), \lambda_t(v)\}&=&\delta(u-v)\\
\{\pi_{\lambda_u}(u), \lambda_u(v)\}&=&\delta(u-v). \label{lamPB}
\eeqa
Adapting the analysis of constraints to the present case, we find the secondary constraints 
\beqa
\mathcal{G}_{\lambda_u}&=& -\beta_i E^i =-\beta_3 \\
\mathcal{G}_{\lambda_t}&=& -\pi_i E^i =-\pi_3 \\
\mathcal{G}_{\beta_i}&=&J^i -\Lambda^{ij}\pi_j -\lambda_u \delta_3^i  \eeqa
whose algebra yields
\beqa
\{\mathcal{G}_{\beta_a}(u),\mathcal{G}_{\beta_b}(v)\}&=&\epsilon_{ab}\mathcal{G}_{\lambda_t}(u)\delta(u-v) \\
\{\mathcal{G}_{\lambda_t}(u),\mathcal{G}_{\beta_a}(v)\}&=&\epsilon_{ab}\mathcal{G}_{\beta_b}(u)\delta(u-v)\\
\{\mathcal{G}_{\lambda_t}(u),\mathcal{G}_{\beta_3}(v)\}&=&-\delta'(u-v)\\
\{\mathcal{G}_{\beta_a}(u),\mathcal{G}_{\beta_3}(v)\}&=&J^a(u)\delta(u-v)\\
\{\mathcal{G}_{\beta_i}(u),\pi_{\lambda_u}(v)\}&=&\delta_{i3}\delta(u-v)\\
\{\mathcal{G}_{\lambda_u}(u),\pi_{\beta_i}(v)\}&=&\delta_{i3}\delta(u-v)
\eeqa
all others being zero. Therefore, from imposing the conservation of secondary constraints, we obtain
\beqa
\dot{ \mathcal{G}}_{\lambda_u}&=&a_{\beta_3}\\
\dot{ \mathcal{G}}_{\beta_a}&=&\beta_3 J^a-\epsilon_{ab}(\beta_b { \mathcal{G}}_{\lambda_t}+ \lambda_t {\mathcal{G}}_{\beta_b})\\
\dot{ \mathcal{G}}_{\beta_3}&=& \beta_a J^a- a_u + \del_u \lambda_t 
\eeqa
yielding
\be\label{fixedlagmul}
a_{\beta_3}=\beta_3=0; \;\;\;  a_u= \beta_a J^a+\del_u \lambda_t .
\ee
In agreement with the general results of  Sec. \ref{dirco}, we can conclude that, out of the $2n+4=10$ constraints of the model, four of them are second class, namely
\be
\mathcal{G}_{\lambda_u}, \;\; \pi_{\lambda_u}, \;\; \mathcal{G}_{\beta_3},\;\; \pi_{\beta_3}.
\ee
The dynamics is retrieved by the total Hamiltonian $H_1$, with canonical Poisson brackets \eqn{piPB}-\eqn{lamPB}  and some of the Lagrange multipliers fixed by Eq. \eqn{fixedlagmul}
\be
H_1= \int du\; \left[ \beta_a \mathcal{G}_{\beta_a} + \lambda_t \mathcal{G}_{\lambda_t}+ a_u \pi_{\lambda_u} + a_t \pi_{\lambda_t} + a_{\beta_a} \pi_{\beta_a} \right], \;\;\;\; a=1,2.
\ee
It may be easily verified that 
the algebra of gauge generators 
\be\label{gautr}
K(\beta, \lambda_t)= \int du \; \left[ \beta_a \mathcal{G}_{\beta_a} + \lambda_t \mathcal{G}_{\lambda_t}+a_t \pi_{\lambda_t} + a_{\beta_a} \pi_{\beta_a} \right], \;\;\;\; a=1,2.
\ee
closes according to Eq. \eqn{KKrel}.

To close this section we apply the results of \ref{intcont} to the case of $SU(2)$ for the integration of the fields $\eta$ and $\lambda$. The resulting action is here adapted as
\begin{equation}
S_2=-\frac{1}{2}\int_{\Sigma}  \langle d\vartheta ,(g^{-1}dg) \wedge (g^{-1}dg) \rangle,
\end{equation}
and by writing $d\vartheta$ explicitly we have
\begin{equation}
S_2=-\frac{1}{2}\int_{\Sigma} \epsilon_{ab} (g^{-1}dg)^a \wedge (g^{-1}dg)^b=\int_{\Sigma}  d^2u \, \epsilon_{ab} A^a J^b,
\end{equation}
with degenerate $B$-field $B_{ab}=\epsilon_{ab}$, all other components being zero.

\subsection{Integration on locally conformal symplectic manifolds}

Let us now consider a $2n$-dimensional locally conformal symplectic manifold $M$ with the non-degenerate two-form $\omega$ and closed one-form $\alpha$ satisfying \eqn{lcscond}, or equivalently, and especially useful for our purposes, \eqn{lcscond2}. We have from the latter
\be\label{LambdaE}
\Lambda= \omega^{-1},\;\;\; E^i= (\omega^{-1})^{ij} \alpha_j 
\ee
Therefore, by multiplying  \eqn{eomjacobi1} with $(\omega)_{\ell i }$ we arrive at 
\begin{equation}
(\omega)_{\ell i}(dX^i+\Lambda^{ij}\eta_j-E^i \lambda)=\omega_{\ell j} dX^j+\eta_\ell-\alpha_\ell \lambda=0,
\end{equation}
so that $\eta$ can be written as
\begin{equation}\label{lcsaux}
\eta_\ell=-\omega_{\ell j} dX^j+\alpha_\ell \lambda.
\end{equation}
Note that in this case, differently from   contact manifolds, it is not possible to explicitly decouple $\eta$ and $\lambda$. However, on substituting \eqn{lcsaux}, together with the second  of Eqs. \eqn{LambdaE} into the action functional, after a few simple manipulations  it is possible to verify that the terms proportional to $\lambda$ simplify out and we are left with 
\begin{equation}\label{secondorderlcs}
S_2=\int_{\Sigma} \omega_{ij} \, dX^i \wedge dX^j = \int_\Sigma X^*(\omega)
\end{equation}
where we have restored the pull-back map in the second equality.
Note that this is formally of the same form as \eqn{secondordercontact} and of the $A$-model 
but it differs from both cases. 
In particular, the role of the $B$-field is represented  by the two-form $\omega$ which is non-degenerate and it is not closed since it satisfies \eqn{lcscond}, so in this case there is place for fluxes on the target. Obviously, if $\alpha=0$ the manifold $M$ becomes a symplectic manifold and the theory reproduces the original $A$-model as a particular case.

\subsection{LCS manifolds} 
Examples of LCS manifolds may be built, according to  \cite{Vaisman1985}, in the following way. The starting point is a    contact manifold  $(M^{2n-1}, \theta), \,\, n\ge 2$,  with contact form $\vartheta$.
The manifold $(S^1\times M^{2n-1}, \omega)$ is  LCS with non degenerate 2-form $\omega$ given by
\be\label{lcsomega}
\omega= \vartheta\wedge \alpha + d\vartheta
\ee
where $\alpha\in \Omega^1(S^1)$ the volume form on the circle. Therefore we can easily construct an interesting {non-trivial} example by considering the product $S^1\times S^3$, with $S^3$ the  contact manifold associated with the group $SU(2)$ previously described.  The Jacobi structure $(\Lambda, E)$ can be worked out, yielding
\be
\Lambda= \omega^{-1},\;\;\; E=\Lambda(\alpha)
\ee
which, in local coordinates for the circle $S^1$, with $\alpha= d\phi$ becomes
\be
\Lambda = Y_3\wedge \del\phi -Y_1\wedge Y_2,\;\;\; E= -Y_3.
\ee
According to \cite{Vaisman1985}, as a generalisation of the latter, one can consider principal bundles $(P,M^{2n-1}, U(1))$ with basis the  contact manifold $M^{2n-1}$ and structure group $U(1)$. P may be endowed with the LCS structure \eqn{lcsomega}
where $\alpha$ is the volume form of the structure group $U(1)$ and $\vartheta$ a $U(1)$ connection. 
If the curvature of the connection
$
\psi= d\vartheta
$
is such that $\alpha\wedge\vartheta\wedge(\psi)^{n-1} \ne 0$ (namely it defines a volume form on $P$), then $\omega$ is a LCS which is not globally conformally symplectic.

The two models considered in this section are  new to our knowledge; they cannot  be obtained from the Poisson sigma model, unless adding additional degrees of freedom, and fully rely on   the underlying Jacobi geometry of the target. The LCS model is especially interesting with respect to its property of being equivalent to a Lagrangian model on the tangent manifold TPM with a two form which is neither degenerate nor closed.  In next section we shall see a dynamical generalisation of Jacobi sigma models, where this issue will be discussed again.

\section{Dynamical Jacobi}
\label{Secdynamical}

In this section, we review a non-topological extension of the Jacobi sigma model introduced  in \cite{Bascone2021}, which generalises the approach proposed in \cite{Schupp2012} for the Poisson sigma model. As we already briefly discussed in Sec. \ref{Secpoissonsigmamodel}, it is possible to add a simple non-topological term to the Poisson sigma model action, which is just a Casimir function on the target manifold, so that it does not spoil the gauge invariance. However, another modification is possible, which might have interesting string applications,  in which a dynamical term containing both the metric on $\Sigma$ and on $M$ is considered. 

The action for the dynamical model gets modified with respect to the topological action analysed so far, according to:
\begin{equation}
S(X, \eta, \lambda)=\int_{\Sigma} \left[\eta_i \wedge dX^i+\frac{1}{2}\Lambda^{ij}(X) \,\eta_i \wedge \eta_j-E^i(X) \,\eta_i \wedge \lambda +\frac{1}{2}(G^{-1})^{ij}(X) \, \eta_i \wedge \star \eta_j \right],
\end{equation}
where the metric on the worldsheet $\Sigma$, $g=diag(1,-1)$,  is implemented via the Hodge star operator $\star$, while $G$ is a metric tensor on the target Jacobi manifold $M$.

Since $G$ is non-degenerate by definition, this allows us to integrate the auxiliary fields for a generic Jacobi manifold $M$ so to obtain a Polyakov string action for the embedding maps $X$, as we will see. In fact, the new equations of motion are
\begin{equation}\label{eometric1}
dX^i+\Lambda^{ij}\eta_j-E^i \lambda+(G^{-1})^{ij} \star \eta_j=0,
\end{equation}
\begin{equation}\label{eometric2}
d\eta_i+\frac{1}{2}\partial_i \Lambda^{jk}\eta_j \wedge \eta_k-\partial_i E^j \eta_j \wedge \lambda+\frac{1}{2}\partial_i (G^{-1})^{jk}\eta_j \wedge \star \eta_k=0,
\end{equation}
\begin{equation}\label{eometric3}
E^i \eta_i=0.
\end{equation}
Being $G$ naturally non-degenerate we can solve \eqn{eometric1} for  $\star \eta$, 
\begin{equation}\label{eqstareta}
\star \eta_j=-G_{ij}\left(dX^i+\Lambda^{ik}\eta_k-E^i \lambda \right)
\end{equation}
and obtain $\eta$ by applying   the Hodge star to the latter 
\begin{equation}\label{eqeta}
\eta_p=-{(M^{-1})^j}_p G_{ij}\left(\star dX^i-\Lambda^{ik}G_{\ell k} dX^{\ell}+\Lambda^{ik}G_{\ell k}E^{\ell}\lambda-E^i \star \lambda \right),
\end{equation}
with  ${M^p}_j={\delta^p}_j-G_{j i}\Lambda^{ik}G_{k \ell}\Lambda^{\ell p}$  a   symmetric matrix, which we may assume to be non-degenerate irrespective of the  rank  of  $\Lambda$. 
The action becomes then 
\begin{equation}\label{actionlambda}
\begin{aligned}
S(X, \lambda) = \int_{\Sigma} {} & \bigg[\frac{1}{2}{(M^{-1})^p}_i G_{jp} \, dX^i \wedge \star dX^j- \frac{1}{2}{(M^{-1})^p}_i G_{\ell p} \Lambda^{\ell k} G_{jk} \, dX^i \wedge dX^j \\ & - \frac{1}{2}{(M^{-1})^p}_i G_{\ell p}\Lambda^{\ell k} G_{mk} E^m \lambda \wedge dX^i+ \frac{1}{2}{(M^{-1})^p}_i G_{\ell p} E^{\ell} \star \lambda \wedge dX^i \bigg].
\end{aligned}
\end{equation}
In order to integrate out the remaining auxiliary field,  $\lambda$, 
we use  the inner product on the space of one-forms, 
\begin{equation}
\int_{\Sigma} \star \lambda \wedge dX=-\int \lambda \wedge \star dX,
\end{equation}
so that  \eqn{actionlambda} $\lambda$ becomes nothing more than a Lagrange multiplier imposing the geometric constraint
\begin{equation}\label{constraintpolyakov}
(M^{-1})_{i \ell}\left( \Lambda^{\ell k}G_{mk}E^m dX^i+E^{\ell} \star dX^i\right)=0.
\end{equation}
This finally leads  to the result that  the term proportional to  $\lambda$ vanishes on-shell and we are left with the second order action
\begin{equation}\label{PolJac}
S=\int_{\Sigma} \left[g_{ij} dX^i \wedge \star dX^j+B_{ij} dX^i \wedge dX^j \right]
\end{equation}
where the  metric $g$ and the $B$-field are defined in terms of  $G$ and $M$ according to:
\begin{equation}\label{metricbfield}
g_{ij}=G_{jp} {(M^{-1})^p}_i, \quad B_{ij}=G_{ik}{(M^{-1})^p}_j G_{ p \ell}\Lambda^{\ell k}.
\end{equation}
To summarise, we have obtained  a Polyakov string action, with target space a Jacobi manifold, represented by Eq. \eqn{PolJac}. The  Jacobi bi-vector field $\Lambda$ enters   the definition of the metric and the B-field, while the Reeb vector field $E$  is part of   the constraint equation \eqn{constraintpolyakov}.

\subsection{Dynamical model on $SU(2)$}

To give an example of the Polyakov action obtained in \eqn{PolJac} we consider again the $SU(2)$ group manifold as target, so to obtain the dynamical completion to the topological model already considered in Sec \ref{topjacsu2}. In particular, as a metric tensor on the target we introduce the natural Cartan--Killing metric on $SU(2)$: $G_{ij}=\delta_{ij}$. By using $G_{ij}=\delta_{ij}$ and $\Lambda^{ij}=\epsilon^{3ij}$, the metric $g$ and $B$-field are then obtained from \eqn{metricbfield} as 
\begin{equation}\label{backgsu2}
g_{ij} = h_{ij}= \delta_{ij}-\frac{1}{2}\epsilon_{ik3}\delta^{kl}\epsilon_{jl3}, \quad B_{ij}=-\frac{1}{2}\epsilon_{3ij},
\end{equation}
so to have 
\begin{equation}
S=\int_{\Sigma} \left[h_{ij} (g^{-1}dg)^i \wedge \star (g^{-1}dg)^j-\frac{1}{2}\epsilon_{3ij} (g^{-1}dg)^i \wedge (g^{-1}dg)^j \right].
\end{equation}
From the analysis of the previous section we know that this action has to be complemented with the geometric constraint in Eq. (\ref{constraintpolyakov}), i.e. in this case
\be
(g^{-1}dg)^3=0.
\ee 
It is interesting to note the form of the background metric $h$ in \eqn{backgsu2}. This metric has been already obtained in the context of Poisson--Lie duality of $SU(2)$ sigma models \cite{Marotta2019, Bascone2020, Bascone2020a, PV19, Marotta2018, Pezzella2019} as a non-degenerate metric for the group manifold of $SB(2, \mathbb{C})$, the Borel subgroup of $SL(2,\mathbb{C})$ of upper triangular matrices with complex elements with real diagonal and unit determinant.  The latter plays the role of the Poisson-Lie dual of $SU(2)$ in the Manin triple decomposition of the group $SL(2,\mathbb{C})$.  Therefore, it is an interesting question to understand the possible relation between the two models. Interestingly, Poisson-Lie groups are discussed in \cite{deLeon1997} in relation with Jacobi structures. We plan to come back to  this question in future investigations.

\section{Discussion}
Let us summarise the main aspects of the model. The Jacobi sigma model is a generalisation of the well known Poisson sigma model.  It  is a two-dimensional topological non-linear gauge theory describing strings moving on a Jacobi background. 
It   can be related to a field theory with a higher dimensional target, which is  a Poisson sigma model for the `Poissonised' manifold $M\times \R$. The so called  Poissonisation procedure consists in the construction of a homogeneous Poisson structure on  $M \times \mathbb{R}$ from the Jacobi structure on the Jacobi manifold $M$. The two models may be seen to yield the same dynamics, after reduction, provided we impose extra constraints at the boundary.

We  have analysed the canonical formulation of the model, which exhibits first and second class constraints, with the former  generating gauge transformations. Interestingly, it is possible to establish an homomorphism between the algebra of gauge transformations and the algebra of sections of the 1-jet bundle $J^1M$, which generalises an analogous result for the Poisson sigma model, where the role of $J^1M$  is played by $T^* M$. The reduced phase space of the model, which is obtained  as the  manifold of constraints modulo gauge symmetries,  has finite dimension, equal to $2 \, \text{dim}M-2 $.

Two main classes of target spaces have been explicitly considered, namely  contact and locally conformal symplectic manifolds. We have shown that in both cases the auxiliary fields can be integrated out and a second-order action description in terms  of the sole embedding maps can be given. In the case of the Poisson sigma model, this  is only possible if the target manifold is  symplectic, so that the Poisson bi-vector can be inverted;  in such a case  the resulting theory is that of a $A$-model and the $B$-field is the symplectic two-form.
For the models at hand we obtain different results: 
on contact manifolds the resulting $B$-field is the exterior derivative of the contact one-form, which   is closed but  degenerate, while for the locally conformal symplectic manifolds the $B$-field is the LCS two-form $\omega$  which is neither degenerate nor closed, allowing for the possibility of generating  fluxes without the need to twist the model. A similar situation occurs for dynamical models (cfr. Eq. \eqn{metricbfield}). In view of the importance of fluxes in relation with  string compactification,  the occurrence of two-forms which are not closed in the context of LCS manifolds is, therefore, interesting and needs to be further investigated.  The original $A$-model of string theory is naturally recovered from the locally conformal symplectic case when the one-form is identically vanishing. The group manifold of $SU(2)$ has been considered as an explicit example of contact manifold. As for interesting examples of  LCS manifolds, we have shortly reviewed a constructive procedure due to Vaisman and shown that the manifold $SU(2)\times U(1)$ may be endowed with a  Jacobi structure. Examples of Jacobi manifolds which are neither contact nor LCS may be found in \cite{deLeon1997}. They include dual algebras of Poisson--Lie groups, which we think could be of interest in the context of Poisson--Lie T--Duality. We plan to address the problem in future work.

Finally, we have reviewed   a dynamical extension of the model, which is obtained  by adding  a metric term to the action functional.  On integrating out   the auxiliary fields,  a Polyakov action is obtained, with a   metric and $B$-field,  which are  explicitly written in terms of  the Jacobi bi--vector field $\Lambda$. The model is  supplemented by  a geometric constraint which is related to the Reeb vector field.

Future directions of research include the quantisation of the model, its relation with Poisson-Lie symmetry and duality and the groupoid structure of the reduced phase space. Moreover,  the possibility of having  non-closed B-fields in the context of LCS manifolds, both for the topological and dynamical models, shall be further investigated.


\noindent{\bf Acknowledgements}The authors are especially indebted with Fabio Di Cosmo for letting them know the work of Vaisman \cite{Vaisman2000} where the Koszul bracket is extended to Jacobi manifolds. Fruitful  discussions with Ivano Basile, Goffredo Chirco and Florio Ciaglia are also  acknowledged, which helped clarifying  many aspects of the model.

\end{document}